\documentclass{article}

\usepackage{arxiv}
\usepackage[utf8]{inputenc} 
\usepackage[T1]{fontenc}    
\usepackage{hyperref}       
\usepackage{url}            
\usepackage{booktabs}       
\usepackage{amsfonts}       
\usepackage{nicefrac}       
\usepackage{microtype}      
\usepackage{lipsum}		
\usepackage{graphicx}
\usepackage{natbib}
\usepackage{doi}

\usepackage[labelfont=bf]{caption}

\usepackage{amsmath,amsfonts,amssymb,amsthm,booktabs,color,epsfig,graphicx,hyperref,url}
\usepackage{subcaption}

\theoremstyle{plain}

\newtheorem{proposition}{Proposition}

\theoremstyle{definition}

\providecommand{\bo}{\mathbf} 
\providecommand{\bs}{\boldsymbol}
\providecommand{\E}{\mathbb{E}}
\providecommand{\cov}{\mathrm{COV}}
\providecommand{\diag}{\mathrm{diag}}

\providecommand{\covAxis}{\mathrm{COVAXIS}}
\providecommand{\tcov}{\mathrm{TCOV}}
\providecommand{\mcd}{\mathrm{MCD}}
\newcommand{\pkg}[1]{\texttt{#1}}
\newcommand{\proglang}[1]{\textsf{#1}}

\usepackage{xcolor}

\title{Invariant Coordinate Selection and Fisher discriminant subspace \\ beyond the case of two groups}

\author{ \href{https://orcid.org/0009-0003-0790-3720}{\includegraphics[scale=0.06]{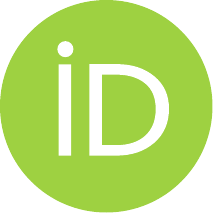}\hspace{1mm}Colombe Becquart} \\
Toulouse School of Economics  \\ France \\
Université de Toulouse \\
 France \\
	\And
	\href{https://orcid.org/0000-0002-6511-9091}{\includegraphics[scale=0.06]{orcid.pdf}\hspace{1mm}Aurore Archimbaud}\thanks{Corresponding author: a.archimbaud@tbs-education.fr} \\
	TBS Business School \\ France \\
	\And
    \href{https://orcid.org/0000-0001-8970-8061}{\includegraphics[scale=0.06]{orcid.pdf}\hspace{1mm}Anne Ruiz-Gazen} \\
	Toulouse School of Economics \\
 France \\
    \And
    {\includegraphics[scale=0.06]{orcid.pdf}\hspace{1mm}Luka Pril\'{c}} \\
	Toulouse School of Economics \\
   France \\
    \And
    \href{https://orcid.org/0000-0002-3758-8501}{\includegraphics[scale=0.06]{orcid.pdf}\hspace{1mm}Klaus Nordhausen} \\
	University of Helsinki\\ Finland \\University of Jyv\"askyl\"a\ \\ Finland \\
}

\renewcommand{\shorttitle}{ICS and Fisher discriminant subspace beyond the case of two groups}

\begin{document}

\maketitle

\begin{abstract}
Invariant Coordinate Selection (ICS) is a multivariate technique that relies on the simultaneous diagonalization of two scatter matrices. It serves various purposes, including its use as a dimension reduction tool prior to clustering or outlier detection. 
ICS's theoretical foundation establishes why and when the identified subspace should contain relevant information by demonstrating its connection with the Fisher discriminant subspace (FDS). These general results have been examined in detail primarily for specific scatter combinations within a two-cluster framework. In this study, we expand these investigations to include more clusters and scatter combinations. Our analysis reveals the importance of distinguishing whether the group centers matrix has full rank. In the full-rank case, we establish deeper connections between ICS and FDS. We provide a detailed study of these relationships for three clusters when the group centers matrix has full rank and when it does not. Based on these expanded theoretical insights and supported by numerical studies, we conclude that ICS is indeed suitable for recovering the FDS under very general settings and cases of failure seem rare. 

\end{abstract}

\keywords{Dimension reduction \and
Mixture of elliptical distributions \and
Scatter matrix \and
Simultaneous diagonalization \and
Subspace estimation}


\section{Introduction}\label{sec:intro}

In many fields, the number of variables is increasing while it is often assumed that the actual information of interest remains contained within a low-dimensional subspace. This is the fundamental idea behind dimension reduction (DR): that it is possible to estimate this lower-dimensional space without losing crucial information, and that analyzing data within this subspace simplifies the process. Clustering and outlier detection are two unsupervised multivariate methods that can benefit significantly from prior dimension reduction. However, the justifications for why a specific DR method is suitable to recover an effective subspace for clustering and outlier detection are often heuristic in nature. Typical dimension reduction methods include principal component analysis (PCA, \cite{hotelling_analysis_1933,joliffe_principal_2002}), projection pursuit (PP, \cite{huber_projection_1985,jones_what_1987,fischer_repplab_2021}), and invariant coordinate selection (ICS, \cite{tyler_invariant_2009, caussinus_interesting_1990}).

Recent research has sought to provide justifications for using these methods as preprocessing steps, often considering specific settings. The most popular framework involves the Gaussian mixture model, with the benchmark being whether the DR method can estimate Fisher discriminant subspace (FDS) \cite{fisher_use_1936} in an unsupervized, or blind, manner. Outlier detection approaches are incorporated into this framework when some clusters are rare.

The use of PCA in the context of clustering, also known as tandem PCA \cite{arabie_cluster_1994}, can be justified using its relationship to principle points as pointed out in \cite{TarpeyLiFlury1995,MatsuuraKurata2011}. PCA will, however, only rarely recover the FDS as shown in  \cite{radojicic_large-sample_2021}.
PP is considered in this framework when skewness or kurtosis are used as PP indices, and conditions under which these estimate the FDS are detailed in \cite{radojicic_large-sample_2021,RadjicicNordhausenVirta2025} in the two-cluster case.
ICS is the most recent DR method among those mentioned. ICS simultaneously diagonalizes two scatter matrices and has been considered for outlier detection in \cite{archimbaud_ics_2018} and for tandem clustering in \cite{alfons_tandem_2024}. In \cite{tyler_invariant_2009}, the authors actually provide very general results for ICS with respect to FDS estimation. However, detailed results exist only for very specific scatter combinations in very specific settings. The goal of this paper is to extend these results to broader settings. The structure of the paper is as follows. In Section~\ref{sec:icsFDS} we recall ICS in more detail, particularly focusing on its application in estimating FDS. Section~\ref{sec:iraccen} examines the properties of ICS for DR when the clusters in the mixtures exhibit no variability. 
In this context, and when the dimension of the FDS is equal to the number of groups minus one, we prove that for any scatter combination, the ICS eigenvalues do not depend on the cluster locations but only on the cluster proportions. 
Section~\ref{sec:dprop} then gives closer attention to this model for a specific scatter combination and studies the behavior of the ICS eigenvalues when the cluster proportions vary. Section~\ref{sec:aligned} explores the same scatter combination within a Gaussian mixture model consisting of three clusters with a group centers matrix of rank one while the maximum rank is two.
Through a simulation study and empirical application (Sections~\ref{sec:emp_study} and \ref{sec:application}), we validate our theoretical results and examine the behavior of 
previously unstudied scatter combinations in settings studied in the previous sections.
The paper concludes with Section~\ref{sec:conclusion}. Calculations and proofs are provided in the appendix while some details on the parameters for our experiments are provided in the supplementary material.

\section{What is already known about ICS in relation with the Fisher discriminant subspace}\label{sec:icsFDS}

\subsection{General principle of ICS}\label{subsec:ics}

Following the notations in \cite{tyler_invariant_2009}, let $\bo{Y}$ be a random vector of dimension $p$, with distribution function $F_{\bo{Y}}$. Let $\mathcal{P}_p$ be the set of all positive definite symmetric matrices of order $p$. A scatter matrix of $\bo{Y}$, denoted by $\bo{V}(F_{\bo{Y}})$, is a function of the distribution of $\bo{Y}$ that is uniquely defined at $F_{\bo{Y}}$, is such that $\bo{V}(F_{\bo{Y}})$ belongs to $\mathcal{P}_p$, and is affine equivariant, i.e., $\bo{V}(F_{\bo{AY} + b}) = \bo{A V}(F_{\bo{Y}}) \bo{A}^\top$ for all non-singular  matrices $\bo{A}$ of dimension $p \times p$ and for all $b \in \mathbb{R}^p$, and where $^\top$ denotes the transpose operation. In what follows, we will drop the dependence on $F_{\bo{Y}}$ and simply denote by $\bo{V}$ the scatter $\bo{V}(F_{\bo{Y}})$ when the context is obvious. 

If the data distribution is elliptical, then all affine equivariant scatter matrices are proportional \cite{nordhausen_cautionary_2015}. This property is not true in general, outside the context of elliptical distributions, and ICS exploits the difference between two scatter matrices to detect non-elliptical structures such as clusters \citep{alfons_tandem_2024} or outliers \citep{archimbaud_ics_2018}. 
To perform this comparison, ICS relies on the simultaneous diagonalization of two scatter matrices $\bo{V}_1$ and $\bo{V}_2$~: 
\[
\bo{H}^\top \bo{V}_1\bo{H} = \bo{D}_1 \quad \mbox{and} \quad
\bo{H}^\top \bo{V}_2\bo{H} = \bo{D}_2,
\]
where $\bo{D}_1$ and $\bo{D}_2$ are diagonal matrices such that $\bo{D}_1^{-1}\bo{D}_2 = \diag(\rho_1, \ldots, \rho_p)$, $\rho_1\geq \cdots \geq \rho_p$ being the eigenvalues of $\bo{V}_1^{-1}\bo{V}_2$ sorted in descending order, and $\bo{H} = (\bo{h}_1, \ldots, \bo{h}_p)$ is a $p\times p$ non singular matrix containing the corresponding eigenvectors. Usually $\bo{H}$ is scaled such that $\bo{D}_1 = \bo{I}_p$, where $\bo{I}_p$ denotes the identity matrix of dimension $p$.
The term ``generalized eigendecomposition" (and correspondingly, ``generalized eigenvalue" and ``generalized eigenvector") is sometimes used in this context, but for simplicity, we avoid this terminology in the present paper.

The $\bo{Z} = \bo{H}^\top \bo{Y}$ transformation of $\bo{Y}$ leads to new variables, that are invariant under affine transformation in the sense of Theorems 1 and 2 in \cite{tyler_invariant_2009}, and are called invariant coordinates or components.
Apart from being invariant by affine transformations, ICS has many applications and useful properties, as investigated for example in \cite{nordhausen_tools_2008,pena_eigenvectors_2010,nordhausen_multivariate_2011,loperfido_skewness_2013,alashwali_use_2016,loperfido_theoretical_2021,nordhausen_usage_2022,archimbaud_numerical_2023}. Of particular interest to us, is the following property for finite mixture models of the form: 

\begin{equation}\label{model0}
f_{\bo{Y}}(\bo{y})=\mbox{det}(\bs{\Gamma})^{-1/2} \sum_{j=1}^k \alpha_j g_j((\bo{y}-\bs{\mu}_j)^\top\bs{\Gamma}^{-1} (\bo{y}-\bs{\mu}_j)),
\end{equation}
where $\alpha_j>0$ for $j\in\{1,\ldots, k\}$, with $\sum_{j=1}^k\alpha_j=1$, are the mixture proportions, the $\bs{\mu}_j$ are distinct location vectors (also called group centers or group means), $\bs{\Gamma} \in \mathcal{P}_p $ is the within-group scatter matrix parameter, and $g_1,\ldots, g_k$ are non-negative functions. 
Thus, each of the $k$ components of the mixture is elliptical and the standard Gaussian mixture model is a special case.
For such a model, Theorem~4 from \cite{tyler_invariant_2009} is fundamental and justifies the use of ICS as a dimension reduction method when the objective is clustering or anomaly detection.
In short, this result proves that for model~\eqref{model0} and under some conditions, the ICS components associated with the largest and/or smallest eigenvalues span the FDS which is the one obtained by using the linear discriminant analysis in a supervized context. In the next subsection, we recall this theorem and also discuss some of its limitations.

 \subsection{ICS and the Fisher discriminant subspace}\label{subsec:fisher}

In a supervized context, when the distribution of $\bo{Y}$ is a mixture of distributions, \cite{fisher_use_1936} suggested to look for the linear function of the $p$ variables in $\bo{Y}$ which maximizes the ratio of the between-group variability to the within-group variability. This function corresponds to the projection of $\bo{Y}$ onto the eigenvector associated with the largest eigenvalue of the simultaneous decomposition of the between and within-group covariance matrices. In the case of a mixture of two Gaussian groups with means $\bs{\mu}_1$ and $\bs{\mu}_2$ and equal within-group covariance matrices $\bs{\Gamma}$, this eigenvector is equal to $\bs{\Gamma}^{-1}(\bs{\mu}_1-\bs{\mu}_2)$ and spans the FDS with dimension $q=k-1=1$. For a mixture of $k$ groups such as model~\eqref{model0}, the FDS is spanned by the vectors $\bs{\Gamma}^{-1}(\bs{\mu}_j-\bs{\mu}_k)$ for $j\in\{1,\ldots, k-1\}$. 

Let us denote by $q$ the dimension of the vector space associated with the affine space spanned by the group centers $\bs{\mu}_j$, $j\in\{1,\ldots, k\}$.
According to Theorem~4 from \cite{tyler_invariant_2009}, the simultaneous diagonalization of two scatter matrices results in at least one eigenvalue, denoted by $\rho*$, with a multiplicity greater than or equal to $p-q$. This eigenvalue is associated with an eigenspace of dimension at least $p-q$ which is the direct sum of the complementary space of the FDS (with dimension $p-q$) and, if the multiplicity of $\rho*$ is strictly larger than $p-q$, a space that is included in the FDS.  If no eigenvalue has multiplicity greater than $p-q$, the subspace spanned by the eigenvectors associated with the other eigenvalues than $\rho*$ (larger or smaller) is the FDS. 

Theorem~4 in \cite{tyler_invariant_2009} is proven by using the equivariance property of the scatter matrices $\bo{V}_1$ and $\bo{V}_2$, and transforming the random vector $\bo{Y}$, with distribution given by~\eqref{model0}, in the following way.
Let $ \bo{M} = (\bs{\mu}_1, \ldots, \bs{\mu}_k)$ be the matrix of rank~$q$ which contains the location vectors of $\bo{Y}$, and let $ \bo{M}_0 = \bs{\Gamma}^{-\frac{1}{2}}(\bo{M} - \bs{\mu}_k \bo{1}_k^\top) $, where $\bo{1}_k$ is a $k$-dimensional vector of ones.
The QR decomposition of $\bo{M}_0$ is:
$$
\bo{M}_0 = \bo{PT} = \bo{P} \begin{pmatrix} \bo{T}_u & \bo{0} \\ \bo{0} & \bo{0} \\ \end{pmatrix},
$$
where $\bo{P}$ is an orthogonal matrix, $\bo{T} = (\bo{t}_1, \ldots, \bo{t}_{k})$ with $\bo{t}_j$ a $p$-dimensional vector for $j\in\{1,\ldots, k\}$, and $\bo{T}_u$ is an upper triangular matrix of dimension $k-1\geq 1$ such that the last $k-1-q\geq 0$ rows are zero. Note that the dimension of $\bo{T}_u$ in \cite{tyler_invariant_2009} differs slightly from ours, but this does not impact the remainder of the proof.
The distribution of the transformed random vector $\bo{X} = \bo{P}^\top \bs{\Gamma}^{-\frac{1}{2}}(\bo{Y} - \bs{\mu}_k \bo{1}_k^\top)$ is a mixture of spherical distributions with density:
\begin{equation}\label{modelt}
f_{\bo{X}}(\bo{x})= \sum_{j=1}^k \alpha_j g_j((\bo{x}-\bo{t}_j)^\top (\bo{x}-\bo{t}_j)),
\end{equation}
where for $j\in\{1,\ldots, k\}$, $\alpha_j>0$, $\sum_{j=1}^k\alpha_j=1$,  the $\bo{t}_j$ are distinct, $\bo{t}_k$ is the zero vector, and $g_1,\ldots, g_k$ are non-negative functions. By decomposing the scatter matrices $\bo{V}_1(F_{\bo{X}})$ and $\bo{V}_2(F_{\bo{X}})$ in four blocks where the left top block has dimension $q \times q$,  and using the affine equivariance of the scatter matrices, \cite{tyler_invariant_2009} are able to write:
$$\bo{V}_1(F_{\bo{X}})^{-1} \bo{V}_2(F_{\bo{X}}) = \left(\begin{array}{cc}
\bo{A}_q &  \bo{0}\\
\bo{0}     & \gamma \bo{I}_{p-q} 
\end{array}\right),$$
for some $\gamma >0$. This result implies that $\bo{V}_1(F_{\bo{X}})^{-1} \bo{V}_2(F_{\bo{X}})$ has an eigenvalue equal to $\gamma$ with multiplicity at least equal to $p-q$. The associated eigenspace can be written as the direct sum of the complementary space of the FDS (with dimension $p-q$) and, if the multiplicity of $\gamma$ is strictly larger than $p-q$, a space that is included in the FDS. The result is then derived for $\bo{Y}$ by noting that the eigenvalues of $\bo{V}_1(F_{\bo{Y}})^{-1} \bo{V}_2(F_{\bo{Y}})$ are proportional to the ones of $\bo{V}_1(F_{\bo{X}})^{-1} \bo{V}_2(F_{\bo{X}})$. Thus there exists an eigenvalue $\rho*$ of $\bo{V}_1(F_{\bo{Y}})^{-1} \bo{V}_2(F_{\bo{Y}})$ with multiplicity at least $p-q$. Moreover, if the multiplicity of $\rho*$ is $p-q$, it is proven that the eigenvectors of $\bo{V}_1(F_{\bo{Y}})^{-1} \bo{V}_2(F_{\bo{Y}})$ which are not associated with $\rho*$ span the same subspace that is spanned by $\bs{\Gamma}^{-1}(\bo{M}- \bs{\mu}_k \bo{1}_k^\top)$ and which corresponds to the FDS.

This result holds for model~\eqref{model0} and for any pair of scatter matrices, making it very general. The problem, however, is that in order to perform dimension reduction with ICS, the eigenvalues associated with the FDS should be distinct from $\rho*$ meaning that the multiplicity of $\rho*$ should not be greater than $p-q$. In \cite{tyler_invariant_2009}, the authors mention that this condition generally holds except for special cases, and that these special cases depend on the scatter pair and on the model parameters. In the following subsection, we recall some special cases for which the behavior of ICS is understood more precisely.

\subsection{Mixture of two Gaussian groups with different centers and other special cases}\label{subsec:2gr}

If $k=2$, model~\eqref{model0} is a mixture of two elliptical distributions and $q=1$. Thus the multiplicity of $\rho*$ greater than $p-1$ corresponds to a multiplicity equal to $p$ and $\bo{V}_1^{-1} \bo{V}_2$ is proportional to the identity. In this case, it is not possible to distinguish the Fisher discriminant direction from the others. 
In order to better understand the conditions under which this situation occurs, we need to specify the scatter pair and the elliptical distributions in the mixture~\eqref{model0}. The case of a mixture of two Gaussian distributions has been further explored in \cite{archimbaud_ics_2018} in the context of anomaly detection, where the authors recommend the use of the covariance matrix, denoted $\cov$, for $\bo{V}_1$, and the matrix based on fourth-order moments, denoted $\cov_4$, for $\bo{V}_2$:
$$\cov=\E[(\bo{Y} - \E(\bo{Y}))(\bo{Y} - \E(\bo{Y}))^\top], \quad \cov_4=\frac{1}{p+2} \E[d^2(\bo{Y})\, (\bo{Y} - \E(\bo{Y}))(\bo{Y} - \E(\bo{Y}))^\top],$$
where $d^2(\bo{Y})$ is the square of the Mahalanobis distance:
$$d^2(\bo{Y})={(\bo{Y} - \E(\bo{Y}))^\top \cov^{-1}(\bo{Y} - \E(\bo{Y}))}.$$
For the scatter pair combination $\cov-\cov_4$, also known as FOBI \cite{cardoso_source_1989,nordhausen_overview_2019}, $\gamma=1$ and the case where all eigenvalues of ICS are equal to one occurs when one of the groups has a proportion exactly equal to $(3 - \sqrt{3})/6$, i.e., approximately 21\% \citep{tyler_invariant_2009}. If one group has a proportion below this threshold, there is one eigenvalue that is strictly greater than the others, which are all equal to one. Conversely, if both 
groups have proportions above this threshold, there will be an eigenvalue strictly lower than the others, which are all equal to one.
For outlier detection, it makes sense to assume that the proportion of outliers is less than 21\% and thus, selecting the first invariant component is sufficient as proposed in \cite{archimbaud_ics_2018}.
However, limiting the theoretical study to two groups only is restrictive and we aim at finding more general results.

In the appendix of Chapter 2 in \cite{archimbaud_methodes_2018}, the eigenvalues of $\cov^{-1} \cov_4$ are computed for a mixture of three Gaussian distributions: one group with zero mean, and two groups with opposite non-zero mean vectors. The two non-zero group means have the same proportion in the mixture. Due to the symmetry of the mixture, the dimension of the FDS restricts to $q=1$. Corollary 2 in~\cite{archimbaud_methodes_2018} addresses the case where all three groups have a covariance matrix equal to the identity. It states that the largest (or smallest) eigenvalue of ICS corresponds to the FDS if and only if the zero-mean group has a proportion greater than (or less than) $1/3$. In the special case where all groups have equal proportions ($1/3$ each), $\cov^{-1} \cov_4$ becomes the identity matrix, and consequently, ICS fails to detect the group structure. While this three-group case extends beyond the two-group case, the assumption of symmetric groups is highly restrictive. Our goal is to develop more general results.

Before going further into the behavior of ICS for mixture models~\eqref{modelt} with more than two groups, we point out that ICS can also make it possible to find the FDS in mixture models other than~\eqref{modelt}. The so-called ``barrow wheel'' distribution introduced in \cite{hampel_robust_2011} \citep[see also the discussion by Stahel and M\" achler in][]{tyler_invariant_2009} is the following two groups mixture distribution:
$$\left(1-\frac{1}{p}\right){\cal N}\left(\bo{0}_p,\mbox{diag}(\sigma_{11}^2,1,\ldots,1\right)+\frac{1}{p}G,$$
where $p$ is the dimension and
$G$ is such that if $\bo{Y}=(Y_1,\ldots, Y_p)^\top$ is distributed according to $G$, $Y_1^2$ follows a chi-square distribution with $\nu$ degrees of freedom, and is independent of $Y_2$, \ldots, $Y_p$, while $(Y_2,\ldots,Y_p)^\top$ follows a $p-1$ dimensional Gaussian distribution with zero mean and covariance matrix equal to $\sigma_{22}^2 \bo{I}_{p-1}$. As detailed on page 43 of \cite{tyler_invariant_2009}, the result of Theorem~4 still applies to the barrow wheel distribution, for any scatter pair matrices (as soon as they are not proportional). 
Moreover, Proposition 4 in the appendix of Chapter 2 of \cite{archimbaud_methodes_2018} gives the eigenvalues of $\cov^{-1} \cov_4$ as functions of the dimension $p$, the group proportions, and  the parameters $\nu$, $\sigma_{11}^2$ and $\sigma_{22}^2$. 

\subsection{Going further with the scatter pair \texorpdfstring{$\cov-\cov_4$}{} and the Gaussian mixture}\label{subsec:cc4calc}

Considering the scatter pair  $\cov-\cov_4$, it is possible to make calculations for model \eqref{modelt} with Gaussian densities. 
The $\cov$ and $\cov_4$ matrices can be expressed as follows (see Appendix A.1 for details):
\begin{align*}
\cov &= \begin{bmatrix} \bs{\beta} & \bo{0} \\ \bo{0} & \bo{I}_{p-q} \end{bmatrix}, &
\cov^{-1} &= \begin{bmatrix} \bo{B} & \bo{0} \\ \bo{0} & \bo{I}_{p-q} \end{bmatrix}, &
\cov_4 &= \begin{bmatrix} \bs{\Psi} & \bo{0} \\ \bo{0} & \bo{I}_{p-q} \end{bmatrix},
\end{align*}
where $\bs{\beta}=(\beta_{ms})$, $\bs{\Psi} = (\psi_{ms})$ and $\bo{B} = (b_{ms})$ are $q \times q$ matrices, and the product of the two scatter matrices is then:
\begin{equation}\label{covcov4}
    \cov^{-1}\cov_4 = \begin{bmatrix} \bo{B}\bs{\Psi} & \bo{0} \\ \bo{0} & \bo{I}_{p-q} \end{bmatrix}\, .
\end{equation} These computations illustrate Theorem~4 from \cite{tyler_invariant_2009} since we get a block diagonal matrix for $\cov^{-1}\cov_4$ with the identity of dimension $p-q$ as the lower diagonal block, meaning that $\cov^{-1}\cov_4$ has an eigenvalue equal to one with multiplicity at least $p-q$ and the associated eigenspace is in direct sum with the FDS which is the space spanned by the columns of $\bo{T}$ (the group centers matrix).
We go further in Appendix A.1 and derive the terms of matrix $\bs{\Psi}$:
\begin{equation}\label{eq:amn}
    \psi_{ms} =\frac{1}{p+2}\Bigg[\sum_{i=1}^{q}\sum_{j=1}^{q} b_{ij}\E[x^{\mbox{c}}_m x^{\mbox{c}}_s x^{\mbox{c}}_i x^{\mbox{c}}_j] +(p-q) \E[x^{\mbox{c}}_m x^{\mbox{c}}_s]\Bigg],
\end{equation}
for $m,s \in\{1, \ldots, q\}$, with $(x^{\mbox{c}}_1,\ldots,x^{\mbox{c}}_p)^\top=\bo{X}-\E(\bo{X})$. The expectations in $\psi_{ms}$ and the expressions of $\beta_{ms}$ are also provided in Appendix A.1. However, when using~\eqref{covcov4} and \eqref{eq:amn} to compute   $\cov^{-1}\cov_4$, we do not have easy expressions for $\cov^{-1}$ and we propose to compute $\cov^{-1}\cov_4$ and its eigenvalues numerically for particular mixture models~\eqref{modelt}. 

We set the value of $p$ to 6 and the four panels of Fig.~\ref{fig:boxplot_gaussian} correspond to different values of $k$ and $q$. On each panel, we make the group proportions vary on the $x$-axis. 
We consider 20 distinct configurations for the group means. The values are integers from -2000 to 15000, selected to cover a wide range of configurations (see Table S1 and the details in Supplementary Section S1). We then draw boxplots of the eigenvalues of $\cov^{-1}\cov_4$ on Fig.~\ref{fig:boxplot_gaussian}. On most panels, there are $p-q$ eigenvalues equal to one, which is consistent with the block $\bo{I}_{p-q}$ in expression \ref{covcov4}. However, the eigenvalues of the block $\bo{B}\bs{\Psi}$ are more difficult to analyze, and vary with the number of groups $k$, the group centers and the dimension $q$ of the FDS. For the two-group case, we observe the following known results. When no group proportion is less than 21\%, the last eigenvalue is the only eigenvalue distinct from and smaller than one. When one group has a proportion of approximately 21\%, all eigenvalues are equal or nearly equal to one. In the case of a group with a proportion of less than 21\%, only the first eigenvalue differs and is larger than one. 

The number of eigenvalues larger (resp. smaller) than one seems to depend on the group proportions with some variability of the eigenvalues depending on the group means. In all panels, when the groups are balanced and have proportions equal to $1/k$, 
we observe $q$ eigenvalues less than one.
When the group proportions are unbalanced, implying that some groups have proportions of less than $1/k$, we observe that some eigenvalues can be larger than one.
For the two-group case, we know that the value 21\% is a threshold for the group proportions. A change in one of the two group proportions from a value larger than 21\% to a value lower than 21\% leads to an ICS eigenvalue shifting from smaller than one to larger than one. For more than two groups, one question arises. Is there a threshold value, depending on the number of groups, group centers and group proportions, such that, when groups are unbalanced, changing a group proportion from a value larger than the threshold to a value lower than the threshold leads to a change of an eigenvalue smaller than one to an eigenvalue larger than one? The question is not easy to answer.
\begin{figure}[htbp]
    \centering
    \includegraphics[width=1\textwidth]{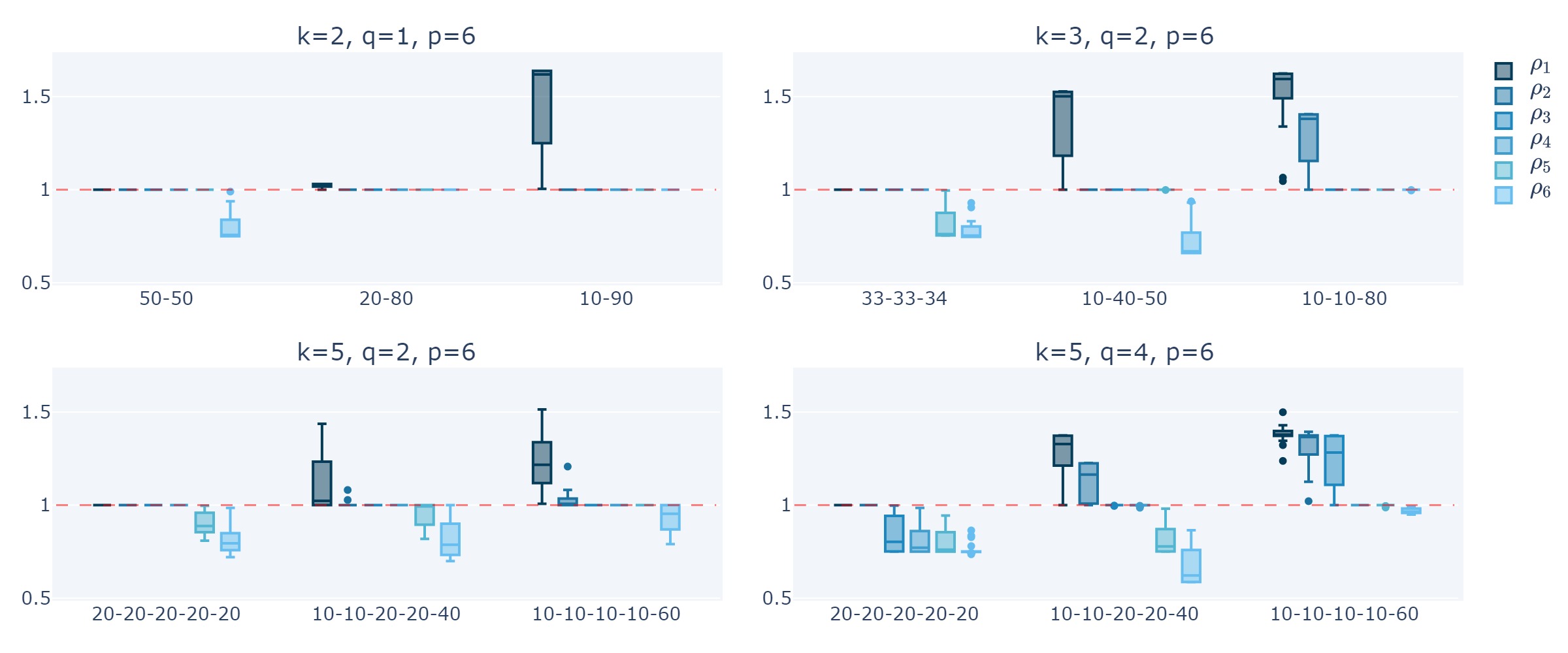}
    \caption{Boxplots of the eigenvalues of $\cov^{-1}\cov_4$, where the group centers vary across 20 different configurations, with $p=6$, and 12 different group proportions scenarios ($\alpha_j$, $j\in\{1,\ldots, k\}$) on the $x$-axis. The values of $k$ and $q$ vary across the panels.}
    \label{fig:boxplot_gaussian}
\end{figure}

Indeed, the experiment illustrates that the eigenvalues of ICS exhibit a complicated dependence on the mixture parameters, requiring a simplified framework to discover interesting results. 
Theorem~4 in \cite{tyler_invariant_2009} ensures that for the elliptical mixture \eqref{model0}, at least $p-q$ eigenvalues of the simultaneous diagonalization of two affine equivariant scatter matrices are equal and are not associated with FDS. With this result established, we now examine the remaining eigenvalues. To simplify the analysis of their behavior, we eliminate the dimensions not associated with FDS, meaning that we focus on the case where $p$ equals $q$. Using a dimension reduction method is not interesting in practice if $p=q$ but the objective here is only to simplify the calculations. Even in this context, it remains difficult to understand the behavior of the ICS eigenvalues associated with the FDS. We therefore further restrict our study to the case where there is no within-group variability. This means the covariance matrix equals the between-group covariance matrix, and the distribution of $\bo{Y}$ is simply a mixture of $k$ Dirac distributions in $p$ dimensions. Finally, in the next two sections, we also focus on the case where $q=k-1$, meaning that the group centers $q \times k$ matrix $\bo{M}$ is of maximum rank $q=k-1$. As explained in the following section, this particular case with $p=q=k-1$ and no within-group variability yields a remarkable result: the ICS eigenvalues depend solely on the proportions of the components of the mixture, and not on the group centers. This property no longer holds when $q$ is less than $k-1$. The latter case is more complex and will be examined only for $k=3$ and $q=1$ in Section \ref{sec:aligned}.

\section{Study of \texorpdfstring{$\bo{V}_1^{-1}\bo{V}_2$}{} eigenvalues when \texorpdfstring{$p=q$}{} with varying group centers, varying group proportions and no within-group variability}\label{sec:iraccen}

Let us consider a mixture of Dirac distributions in $p=q$ dimensions:  $\bo{Y} \sim \sum_{j=1}^{k}\alpha_j \delta_{\bs{\mu}_j}$ with distinct group centers $\bs{\mu}_j$, and  proportions $\alpha_j>0$  such that $\sum_{j=1}^k\alpha_j=1$, for $j\in\{1,\ldots, k\}$. Following \cite{tyler_invariant_2009}, given the affine invariance property of ICS, we can simplify this model and consider the following mixture model:
\begin{equation}\label{modeldc}
    \bo{X} \sim \sum_{j=1}^{k}\alpha_j \delta_{\bo{t}_j},
\end{equation}
with  $\bo{t}_k$ being the $q$-dimensional zero vector, and with distinct group centers $\bo{t}_j$ for $j\in\{1,\ldots, k\}$. This model corresponds in some sense to model~\eqref{modelt} with $p=q$ and after removing the within-group variability. This is an extreme situation but it will give us indications on the behavior of the eigenvalues of ICS when the noise (measured by the within-group covariance matrix) is ``small'' compared to the signal (measured by the between-group covariance matrix). Even if the  Dirac measures do not have probability density functions, it is possible to compute some scatter matrices such as $\cov$ and $\cov_4$, and simultaneously diagonalize them.
Note however that since the between-group covariance is equal to the total covariance matrix for model \eqref{modeldc}, linear discriminant analysis consists in diagonalizing the identity matrix which is of no interest.

\subsection{General theoretical result for \texorpdfstring{$q=k-1$}{} \label{subsec:dcgen}}

For model~\eqref{modeldc} and under the assumption that the group centers matrix is of full rank $q=k-1$, we can derive the following proposition that states that the eigenvalues of ICS do not depend on the group centers. 

\begin{proposition}\label{prop:eigen}
For the mixture model~\eqref{modeldc} with full rank group centers matrix, and for any pair of affine equivariant scatter matrices $\bo{V}_1(F_{\bo{X}})-\bo{V}_2(F_{\bo{X}})$ that exist at $F_{\bo{X}}$, the eigenvalues of ICS, i.e., of $\bo{V}_1(F_{\bo{X}})^{-1}\bo{V}_2(F_{\bo{X}})$, do not depend on the group centers matrix but only on the proportions $\alpha_j$, $j\in\{1,\ldots, k\}$ of the mixture components. 
\end{proposition}

\begin{proof}[\textbf{\upshape Proof:}]
Let us consider two mixtures of the form~\eqref{modeldc} with the same component proportions $\alpha_j$ for $j\in\{1,\ldots, k\}$, but with different $(k-1)\times k$  full rank group center matrices $\bo{T}= (\bo{t}_1, \ldots, \bo{t}_{k})$ and $\bo{\tilde{T}}= (\bo{\tilde{t}}_1, \ldots, \bo{\tilde{t}}_{k})$ respectively, where $\bo{t}_k$ and $\bo{\tilde{t}}_k$ are zero vectors. Let $\bo{T}_{k-1}$ (resp.  $\bo{\tilde{T}}_{k-1}$) be the $(k-1)\times (k-1)$ matrix such that $\bo{T}=(\bo{T}_{k-1},\bo{0})$ (resp. $\bo{\tilde{T}}=(\bo{\tilde{T}}_{k-1},\bo{0})$). The matrices $\bo{T}_{k-1}$ and $\bo{\tilde{T}}_{k-1}$ are square matrices with rank $k-1$ and are thus invertible. We can write $\bo{\tilde{T}}_{k-1}=\bo{A} \bo{T}_{k-1}$ where $\bo{A}=\bo{\tilde{T}}_{k-1}  \bo{T}_{k-1}^{-1}$ is a $(k-1)\times(k-1)$ non-singular matrix and we also have that $\bo{\tilde{T}}=\bo{A}\bo{T}$. Let $\bo{X}$ be a $(k-1)$-dimensional random vector following the mixture~\eqref{modeldc} with group centers matrix $\bo{T}$. The random vector $\bo{Y}=\bo{AX}$ follows the mixture~\eqref{modeldc} with the same proportions $\alpha_i$, $i\in\{1,\ldots, k\}$ and the group centers matrix $\bo{\tilde{T}}$. Using the affine equivariance property of the scatter matrices we have that $\bo{V}_1(F_{\bo{Y}})^{-1}\bo{V}_2(F_{\bo{Y}})=(\bo{A}^\top)^{-1}\bo{V}_1(F_{\bo{X}})^{-1}\bo{V}_2(F_{\bo{X}})\bo{A}^\top$ and has the same eigenvalues as $\bo{V}_1(F_{\bo{X}})^{-1}\bo{V}_2(F_{\bo{X}})$. 
\end{proof}

This result is a kind of generalization of the result for the two-group case detailed in Subsection \ref{subsec:2gr}. For two groups with distinct group centers, we have $q=1$ and thus the assumption $q=k-1$ always holds. The result in Subsection \ref{subsec:2gr} however differs from Proposition \ref{prop:eigen} in several aspects. The result for the two-group case focuses on a Gaussian mixture and accommodates some within-group variability but applies only to the $\cov-\cov_4$ scatter pair. 
In the next subsection, we investigate numerically the ICS eigenvalues for $\cov-\cov_4$ for model~\eqref{modeldc} when $p=q=k-1$.

\subsection{Numerical results for \texorpdfstring{$\cov-\cov_4$}{}}\label{subsec:dcpart}

For model~\eqref{modeldc} with $k=2$, we have $p=q=1$ and it is possible to derive the exact expression of $\cov^{-1}\cov_4$, which is a positive number. More precisely, we get after some simple computation, $\cov=\alpha_1\alpha_2(t_1-t_2)^2$ and $ \cov_4=(\alpha_1^3+\alpha_2^3) (t_1-t_2)^2 /3$ so that
$\cov^{-1}\cov_4=(\alpha_1^3+\alpha_2^3)/(3\alpha_1\alpha_2)$ does not depend on $t_1$ and $t_2$. We then recover a result similar to the one at the beginning of Subsection \ref{subsec:2gr} for $\cov-\cov_4$ and a mixture of two Gaussian distributions: the eigenvalue associated with the FDS is equal to one (which is the value of the other eigenvalues when $p>q$) if and only if one of the two-group proportion is equal to $(3-\sqrt{3})/6$.  More specifically, using the fact that $\alpha_2=1-\alpha_1$ and fixing $\alpha_1 \leq\alpha_2$, $\cov^{-1}\cov_4$ is equal to one (respectively larger or smaller than one) 
if and only if $\alpha_1=(3-\sqrt{3})/6$ (respectively $\alpha_1$ smaller or larger than $(3-\sqrt{3})/6$).

When $k>2$, the eigenvalues of $\cov^{-1}\cov_4$ cannot be easily derived explicitly and numerical computations are necessary. By using the expressions derived in Appendix A.2, alongside numerical computations for the inverse of $\cov$, the eigenvalues of $\cov^{-1}\cov_4$ can be calculated. These calculations are performed for the 20 specified group centers and across 18 different scenarios of group proportions (see Appendix A.2 and Supplementary Section S2 for details). 
Boxplots in Fig.~\ref{fig:boxplot_dirac_q_eq} display the outcomes when $p=q=k-1$, for $k\in\{2, 3, 5, 10\}$ groups (as in Fig.~\ref{fig:boxplot_gaussian}). The boxplots consist of a single line for each scenario, indicating that varying the group means does not impact the eigenvalues when $q=k-1$. However, the values of the eigenvalues differ from one scenario to another. Therefore, it can be concluded that the group proportions do have an impact on the eigenvalues. A more detailed examination of this impact will be presented in the next section.
When $q<k-1$, the results are different, as illustrated in Fig.~\ref{fig:boxplot_dirac_q_lt_1} for $k=3$ and $q=1$, and $k=5$ and $q \in \{1, 2, 3\}$ (see also Fig.~S1 in Supplementary Section~S2 with $k=10$ and $q\in\{1,3,5,7\}$). For each scenario, the eigenvalues exhibit quite large variability across the different group centers. This variability is such that boxplots cross red horizontal dotted line corresponding to the value one. Unlike the case where $q=k-1$, we cannot draw any conclusions about the conditions which cause the eigenvalues to be larger or smaller than one. The influence of group proportions on eigenvalues persists, but the influence of the group centers is also crucial. 
\begin{figure}[htbp]
\centering
\includegraphics[width=1\textwidth]{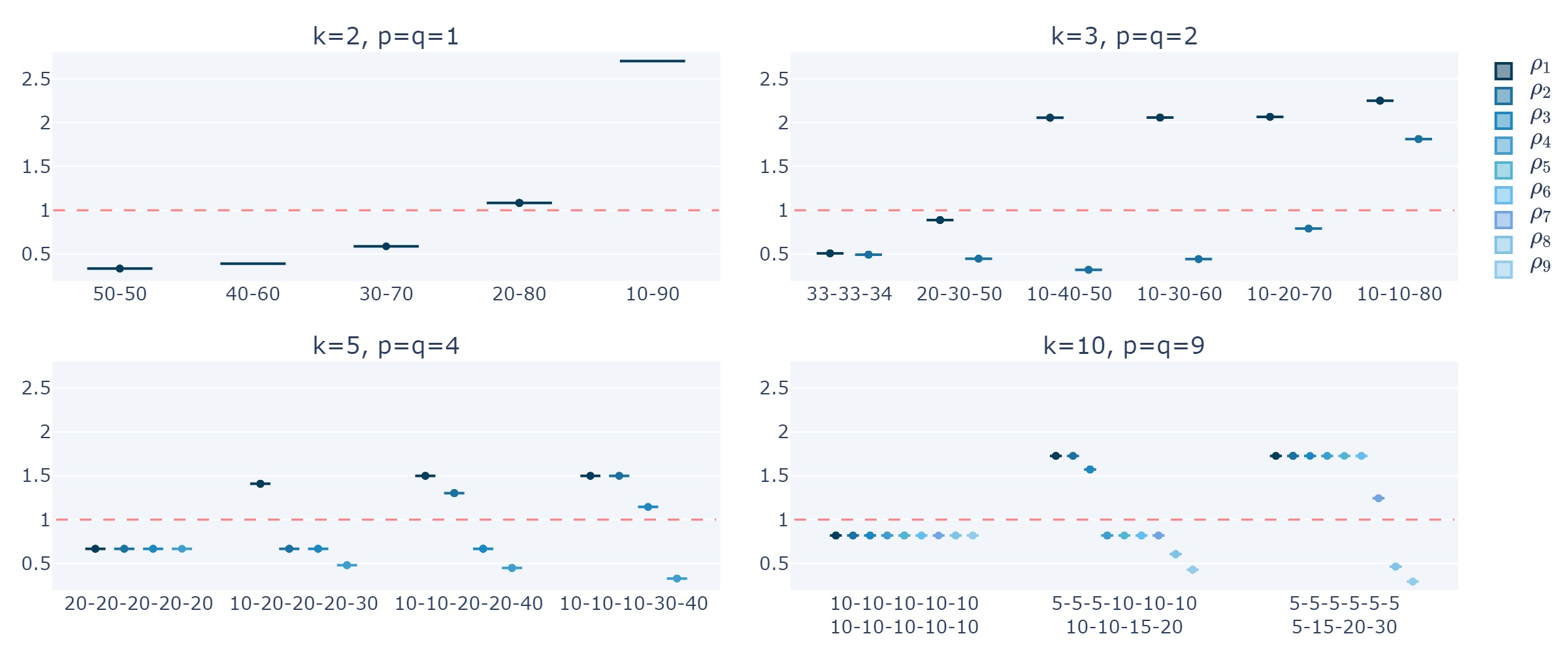}
\caption{Boxplots of the eigenvalues of $\cov^{-1}\cov_4$ for $p = q = k-1$ and with no within-group variability, when the group centers vary across 20 different configurations. The values of $k$ vary across the panels and the $x$-axis of each panel corresponds to 18 group proportions scenarios ($\alpha_j$,~$j\in\{1,\ldots, k\}$). }
\label{fig:boxplot_dirac_q_eq}
\end{figure}
\begin{figure}[htbp]
\centering
\includegraphics[width=1\textwidth]{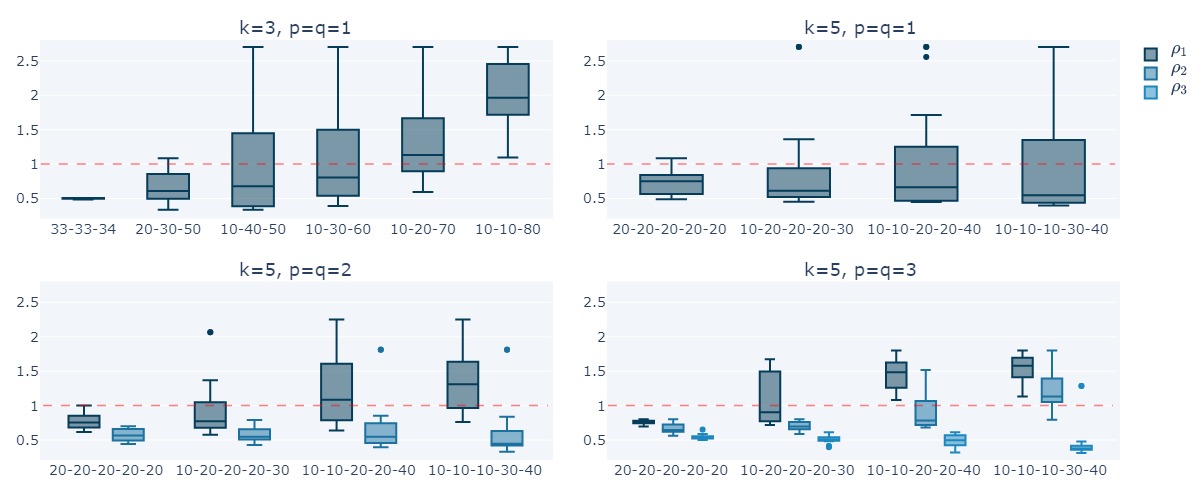}
\caption{Boxplots of the eigenvalues of $\cov^{-1}\cov_4$ 
for $p = q < k-1$ and with no within-group variability, when the group centers vary across 20 different configurations. The values of $k$ and $q$ vary across the panels and the $x$-axis of each panel corresponds to 18 group proportions scenarios ($\alpha_j$, $j\in\{1,\ldots, k\}$).}
\label{fig:boxplot_dirac_q_lt_1}
\end{figure}

\section{Study of \texorpdfstring{$\cov-\cov_4$}{} eigenvalues for  \texorpdfstring{$q=k-1$}{} with fixed group centers, varying group proportions, and no within-group variability.} \label{sec:dprop}

As stated in Section \ref{sec:iraccen}, in the context of a mixture of Dirac distributions with $q=k-1$, only the group proportions influence the eigenvalues of $\cov^{-1}\cov_4$. In this section, the group means are therefore fixed to the first row of the group centers of Table S1 of Supplementary Section S1. The relationship between the group proportions and the eigenvalues is further investigated in the same context as in the previous subsection. 
Theoretical calculations do not permit to measure precisely the impact of group proportions on the eigenvalues of $ \cov^{-1}\cov_4$. Therefore, we vary the group proportions on a grid and calculate numerically the corresponding eigenvalues to better understand their behavior. The three-group case is described in Subsection \ref{subsec:dprop3}, while Subsection \ref{subsec:dpropk} generalises to $k$ groups. 

\subsection{The case of three groups}\label{subsec:dprop3}

One of the most advantageous aspects of studying three groups is the ability to represent the proportions of each group in a ternary diagram, as in Fig.~\ref{fig:tern_grad}. This diagram \citep[see][]{pawlowskyglahn_modeling_2015} displays the values of three positive variables that sum up to 100\%, and can be applied for the three-group proportions. Each point on the ternary diagram represents a scenario of group proportions, and the color of the point is determined by the eigenvalues obtained for that scenario.
As $k=3$, we have $q=2$ and we focus on the case $p=2$ in order to avoid irrelevant dimensions associated with eigenvalues equal to one. $\cov^{-1}\cov_4$ has thus two eigenvalues  $\rho_1 \geq \rho_2$ and two ternary diagrams are plotted: Fig.~\ref{fig:tern_1} for $\rho_1$, and Fig.~\ref{fig:tern_2} for $\rho_2$.

\begin{figure}[htbp]
    \centering
    \begin{subfigure}[b]{0.49\textwidth}
        \centering
        \includegraphics[width=\textwidth]{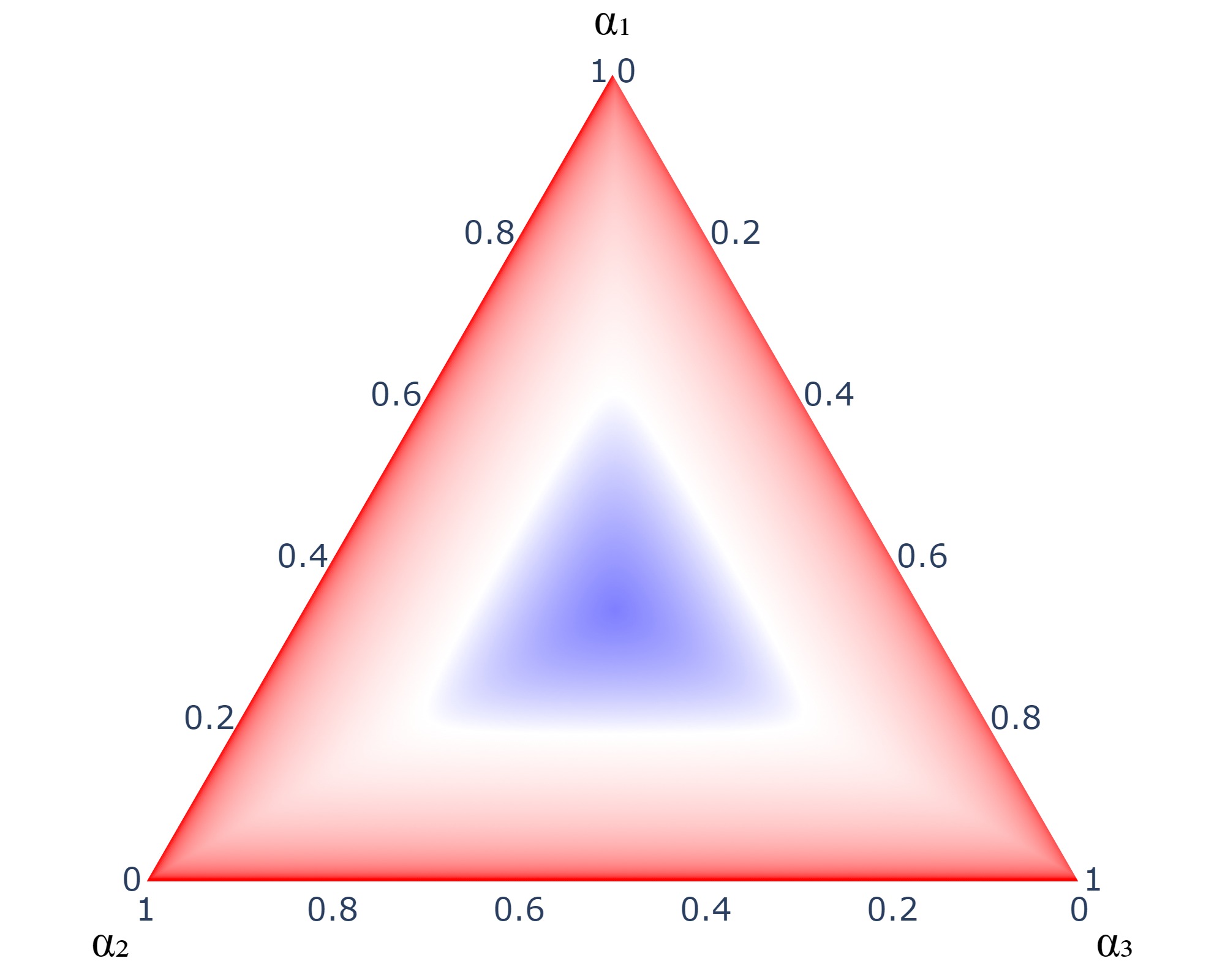}
        \caption{$\rho_1$}
        \label{fig:tern_1}
    \end{subfigure}
    \hfill
    \begin{subfigure}[b]{0.49\textwidth}
        \centering
        \includegraphics[width=\textwidth]{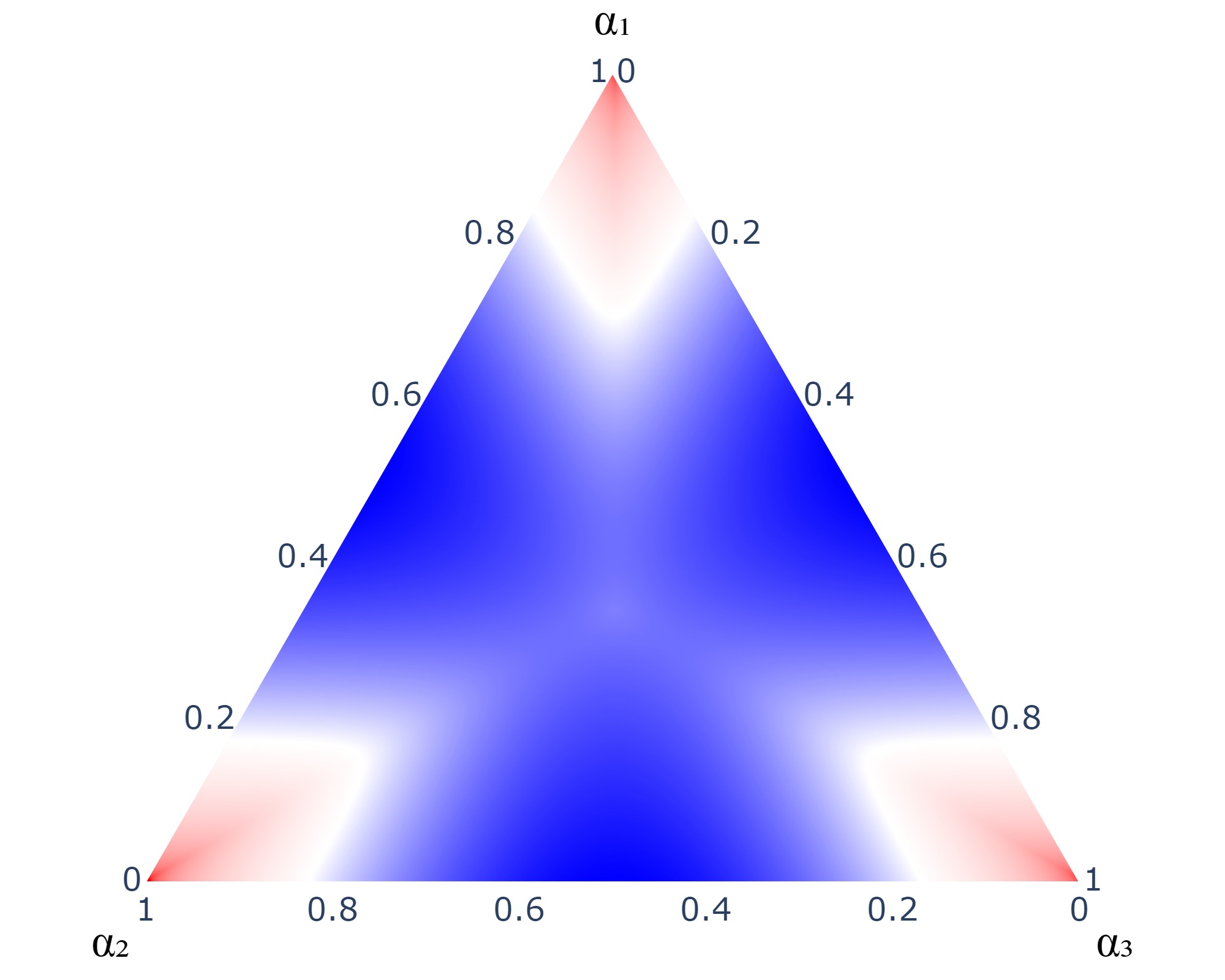}
        \caption{$\rho_2$}
        \label{fig:tern_2}
    \end{subfigure}
    \vskip\baselineskip
    \begin{subfigure}[b]{0.49\textwidth}
        \centering
        \includegraphics[width=\textwidth]{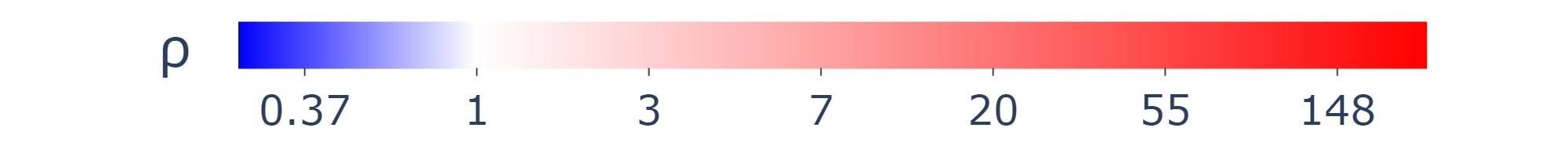}
    \end{subfigure}
    \caption{Ternary diagram of the group proportions ($\alpha_1$, $\alpha_2$, $\alpha_3$) color-coded by the values of $\rho_1$ (a) and $\rho_2$ (b) when $k=3$ and with no within-group variability. The color gradient indicates the departure of the eigenvalues from one (white when equal, blue when smaller and red when larger). Group means are fixed to $\bo{t}_1 = (200, 0)$, $\bo{t}_2 = (400, 100)$, $\bo{t}_3 = (0, 0)$ but have no impact on the results.}
    \label{fig:tern_grad}
\end{figure}

Fig.~\ref{fig:tern_1} illustrates that the values of $\rho_1$ are less than one in the center of the plot, indicating that when none of the three groups has a proportion smaller than roughly 18\%, $\rho_1$ is less than one. Around this region, there is a white area that corresponds (given the continuity of the eigenvalues as functions of the group proportions) to $\rho_1$ almost equal or equal to one.  $\rho_1=1$ is thus indistinguishable from the eigenvalues associated with the $p-q$  dimensions that do not span the FDS. It occurs when one group proportion is roughly 18\%. The area beyond this white zone in the direction of the vertices is red, which indicates that the value of $\rho_1$ is greater than one when a group proportion is smaller than 18\%.
Looking at the areas around the vertices on Fig.~\ref{fig:tern_2} demonstrates that the value of $\rho_2$ is greater than one for scenarios where two proportions are smaller than roughly 18\%. A white area is also present around each of these regions, which corresponds to values of $\rho_2$ almost equal or equal to one. It occurs once again when one group has a proportion of roughly 18\% and another group has a proportion smaller than or equal to 18\%. The rest of the plot is blue, implying $\rho_2$ is smaller than one. A comparison of Fig.~\ref{fig:tern_1} and \ref{fig:tern_2} indicates that the white regions do not seem to intersect, thereby suggesting that the two eigenvalues are not equal to one simultaneously. In the two-group case, at the threshold value of 21\%, all eigenvalues are equal to one. In the three-group case with $q=k-1$, our conjecture is that there is at least one eigenvalue that differs from one for any possible group proportion scenario. 

The eigenvalues representation from Fig.~\ref{fig:tern_1} and \ref{fig:tern_2} are interesting since the greater the deviation from one, the easier it will be to detect the directions that span the FDS when in practice we will have $p>k$.
To facilitate the comparison between the behavior of the two eigenvalues, both are plotted on the same ternary diagram on Fig.~\ref{fig:tern_plot}.  
\begin{figure}[htbp]
    \centering
    \includegraphics[width=0.6\textwidth]{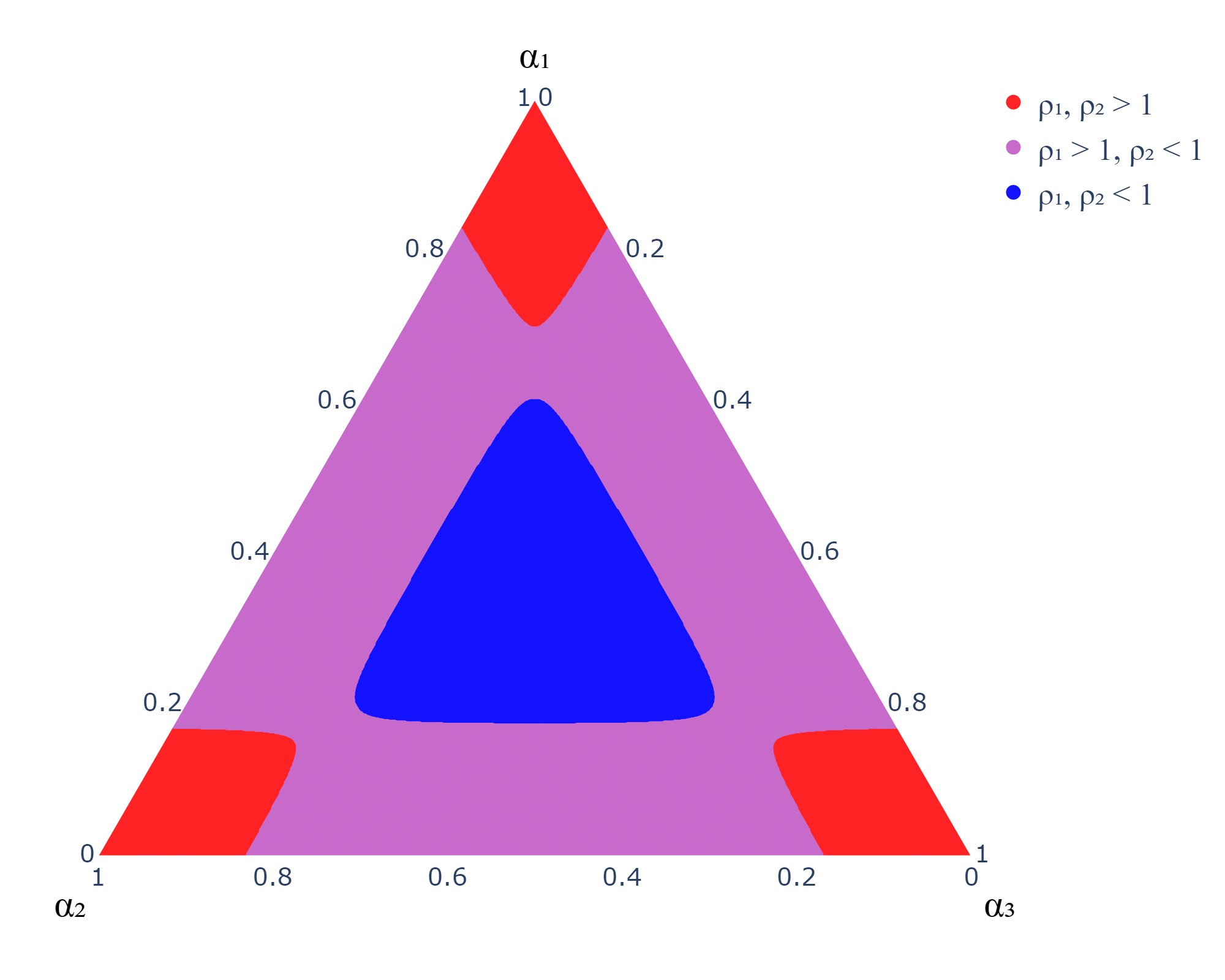}
    \caption{Ternary diagram of the group proportions ($\alpha_1$, $\alpha_2$, $\alpha_3$) color-coded by the eigenvalues of $\cov^{-1}\cov_4$ for $k=3$ and $p=q=2$ (blue when both eigenvalues are smaller than one, purple when one eigenvalue is larger and the other is smaller than one, and red when both eigenvalues are larger than one). Group centers are fixed to $\bo{t}_1 = (200, 0)$, $\bo{t}_2 = (400, 100)$, and $\bo{t}_3 = (0, 0)$ but have no impact on the results.}
    \label{fig:tern_plot}
\end{figure}
In the event that two groups exhibit proportions roughly smaller than 18\%, both eigenvalues are greater than one, as may be observed in the red area of Fig.~\ref{fig:tern_plot}. When a single group proportion is roughly smaller than 18\%, the first eigenvalue is greater than one and the second is less than one. This occurs in the purple region in Fig.~\ref{fig:tern_plot}. Finally, when all groups are in proportion roughly larger than 18\%, in the blue area of Fig.~\ref{fig:tern_plot}, both eigenvalues are less than one. As illustrated in Fig.~\ref{fig:tern_1} and \ref{fig:tern_2}, each boundary between the colored zones should display scenarios for which a single eigenvalue differs from one. However, the analysis has a numerical constraint: the 0.1\% step size for the discretization of the group proportions while we use the exact value of one as the reference eigenvalue. A consequence of this experiment limitation is that only three scenarios exhibit an eigenvalue precisely equal to one. They correspond to a proportion of 60\% for one group and 20\% for the other two groups. These points are not represented in Fig.~\ref{fig:tern_plot}. Refer to Appendix B.1 for additional details on the construction of Fig.~\ref{fig:tern_grad} and \ref{fig:tern_plot}.

The findings for $k=3$ indicate that there is a threshold value (around 18\%) such that
when a group proportion goes from a value above this threshold to a value below, an eigenvalue of $\cov^{-1 }\cov_4$ goes from a value less than one to a value greater than one. For two groups, this threshold is known and is approximately 21\%. From Fig.~\ref{fig:tern_plot}, establishing the precise threshold for three groups is not easy, but it is clear that the threshold is less than 21\% and approximately 18\%. Additionally, the rounded boundaries between the zones suggest that this threshold may slightly vary with the other group proportions. The goal now is to increase the number of groups beyond three, and to study whether there also exists a threshold value for the group proportions that would induce the same impact on the behavior of the ICS eigenvalues with $\cov-\cov_4$.
Subsection \ref{subsec:dpropk} also proposes to numerically approximate the values of these thresholds as a function of the number of groups $k$.

\subsection{The general case of \texorpdfstring{$k$}{} groups}\label{subsec:dpropk}

The thresholds described previously are computed in this section for three setups of group proportions and for values of $k$ between 2 and 10. Note that for a given $k$, the maximum proportion of a group is larger than $1/k$, while the minimum proportion is smaller than $1/k$. The balanced case corresponds to proportions all equal to $1/k$. In Setup~1, we consider scenarios of the mixture proportions such that the first group proportion ranges from 0.001 to $1/k$, the intermediate group proportions are set to $1/k$, and the last group proportion is adjusted to maintain the total sum to one. The eigenvalues of $\cov^{-1}\cov_4$ are calculated for each scenario. The threshold is defined as the first value of the first group proportion for which an eigenvalue exceeds one when decreasing the first group proportion. In Setup~2, the second group proportion is set to the threshold value identified in Setup~1 plus 2\%. The first group proportion ranges from 0.001 to $1/k$. The remaining groups maintain a proportion of $1/k$, with the exception of the last group, whose proportion is adjusted as in Setup~1 to maintain the total sum to one. The objective of this setup is to examine the influence of a group located in proximity to the threshold. It requires at least three groups. The threshold is determined similarly to Setup~1. The grid in Setup~3 is identical to the one in Setup~2, with the exception that the second group proportion is fixed at 5\%, and the last group proportion is adjusted accordingly. This setup includes a group proportion initially below the threshold for any value of $k$ between 3 and 10, thereby ensuring that at least one eigenvalue exceeds one. The threshold is defined as the first value of the first group proportion for which two eigenvalues exceed one. As with Setup~2, this method requires at least three groups. Further details and examples may be found in Appendix B.2. The thresholds identified through the three setups, for $k=2$ to 10 groups, are plotted in Fig.~\ref{fig:thresh}. 
\begin{figure}[htbp]
\centering
\includegraphics[width=0.65\textwidth]{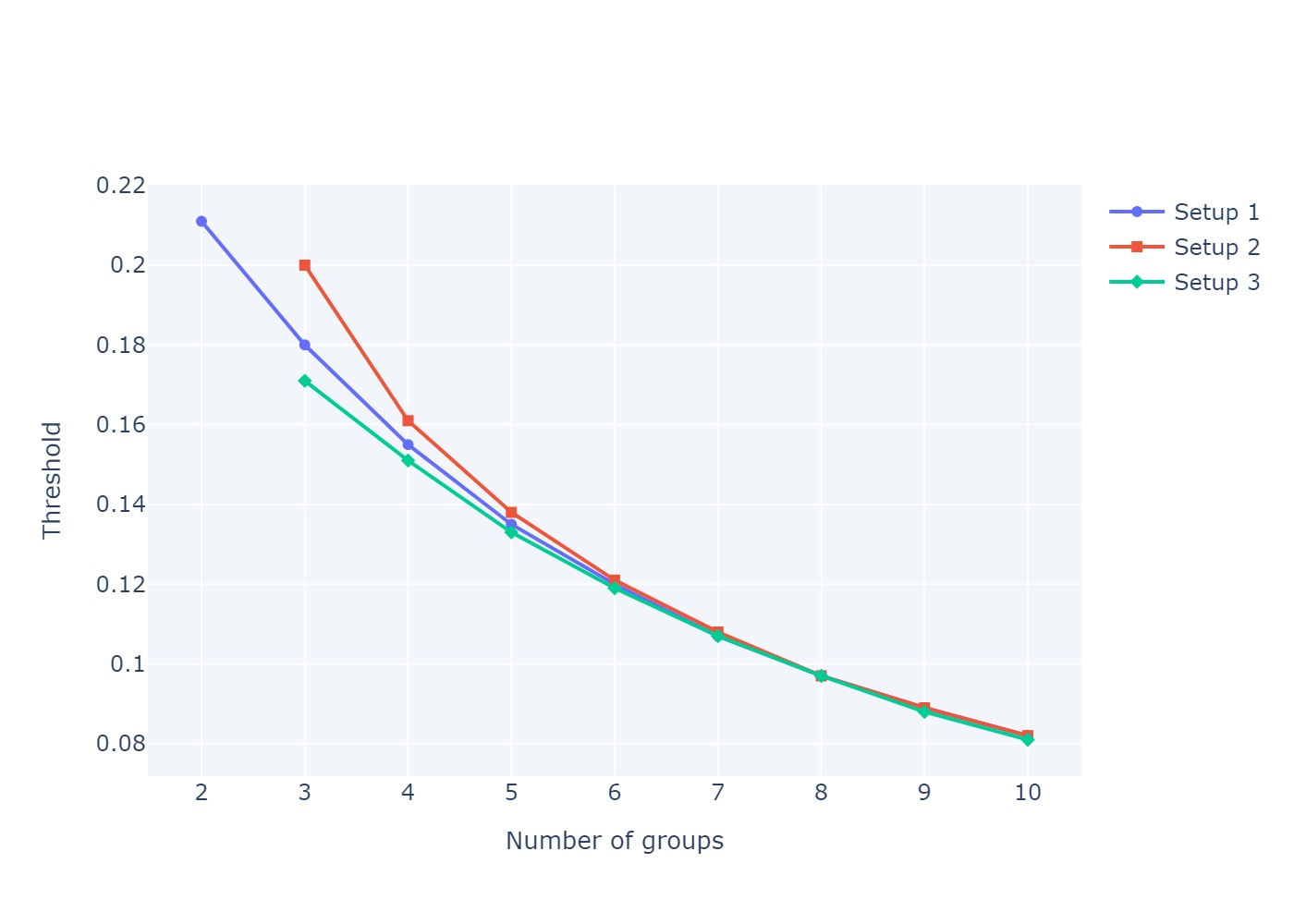}
\caption{Group proportion thresholds computed for different values of the number of groups $k$ and three setups. Setup 1: all group proportions initially equal. Setup 2: $\alpha_2$ near the Setup 1 threshold. Setup 3: $\alpha_2$ equal to 5\%.}
\label{fig:thresh}
\end{figure}
Since the grid generation in Setups 2 and 3 requires more than three groups, there is only one point when $k=2$ for Setup~1. This point is roughly equal to 21\% which is the result already known theoretically. For $k=3$, the threshold value found for Setup~1 is around 18\%, the one for Setup~2 is 20\%, and the one for Setup~3 is 17\%. Fixing a group proportion near the first threshold or below the threshold has indeed an impact. Nonetheless, the thresholds remain around 18\%. As $k$ increases, the thresholds for the three setups appear to converge toward the same value and to depend solely on $k$.
Table \ref{tab:thresholds} in Appendix B.2 gives the specific values of the thresholds for $ k \in \{3, \ldots, 10\}$.
It is worth noting that the presented setups could be applied to a larger number of groups, although we have limited our analysis to 10 groups.

Once the thresholds for three to ten groups have been identified, the boxplots from the previous section can be reproduced for adjusted scenarios. It is of particular interest to visualize the eigenvalues when one group proportion is at the identified threshold, and to select two or three of the previous mixture proportion scenarios to vary the number of groups with proportions below this threshold. Fig.~\ref{fig:boxplot_threshold} illustrates this experiment. Although boxplots are plotted, the group centers remain fixed and equal to the values of the first row of the group centers of Table S1 of Supplementary Section S1.
\begin{figure}[htbp]
\begin{center}
\includegraphics[width=1\textwidth]{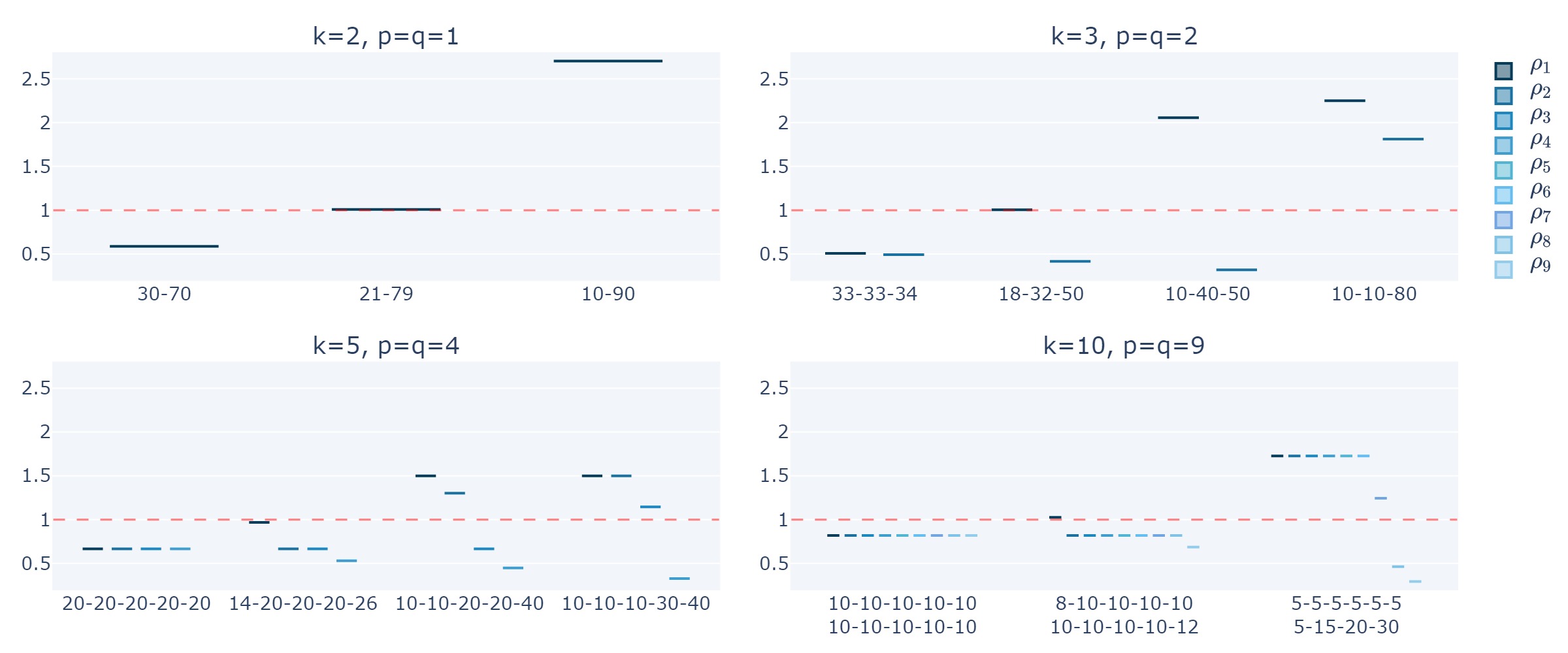}
\caption{Boxplots of the eigenvalues of $\cov^{-1}\cov_4$ with no within-group variability, with $p = q = k-1$, and 14 different group proportions scenarios ($\alpha_j$, $j\in\{1,\ldots, k\}$) on the $x$-axis. The values of $k$ vary across panels. In each panel, the second scenario (on the $x$-axis) contains a group at the identified threshold value: approximately 21\% when $k=2$, 18\% when $k=3$, 14\% when $k=5$, and 8\% when $k=10$}.
\label{fig:boxplot_threshold}
\end{center}
\end{figure}

The results presented in Fig.~\ref{fig:boxplot_threshold} validate the thresholds of Setup~1 presented in Fig.~\ref{fig:thresh} for $k\in\{2, 3, 5, 10\}$ groups. In each boxplot, the scenario where a group is at the threshold level has an eigenvalue almost equal to~one. This is the first eigenvalue, given that these scenarios do not have proportions below the respective threshold. The results are consistent with the logic described in Fig.~\ref{fig:tern_plot}, which states that the number of eigenvalues greater than one is linked to the number of groups with proportion below the threshold. For the scenarios of Fig.~\ref{fig:boxplot_threshold}, when $k=5$ (resp. $k=10$), the number of eigenvalues larger than one is equal to the number of groups with proportions equal to 10\% (resp. 5\%). However, this relationship needs further investigation in cases involving more than three groups.

\section{Study of \texorpdfstring{$\cov^{-1}\cov_4$}{} eigenvalues for a Gaussian mixture of three groups with aligned group centers}\label{sec:aligned}

In Section \ref{sec:iraccen}, Fig.~\ref{fig:boxplot_dirac_q_eq} and \ref{fig:boxplot_dirac_q_lt_1} display the eigenvalues of ICS using $\cov-\cov_4$ for various Dirac mixture models described by~\eqref{modeldc}. Fig.~\ref{fig:boxplot_dirac_q_eq} illustrates that, if $q=k-1$, these eigenvalues do not depend on the group means, and Proposition~\ref{prop:eigen} in Section \ref{sec:iraccen} confirms the result theoretically. However, as shown in Fig.~\ref{fig:boxplot_dirac_q_lt_1}, this result does not hold if $q<k-1$, leading to a more complex analysis of the eigenvalues. In the present section, we examine the eigenvalues of $\cov^{-1}\cov_4$ for model~\eqref{modelt} with three Gaussian groups ($k=3$) and $q=1<k-1=2$. This corresponds to the following $p$-dimensional Gaussian mixture:
\begin{equation}\label{modelk3q1}
    \alpha_1 {\cal N}_p(\bo{t}_1,\bo{I}_p) +  \alpha_2 {\cal N}_p(\bo{t}_2,\bo{I}_p)+  (1-\alpha_1-\alpha_2) {\cal N}_p(\bo{0}_p,\bo{I}_p),
\end{equation}
with $\bo{t}_1=(t_{11},0,\ldots,0)^\top$, $\bo{t}_2=(t_{21}, 0,\ldots,0)^\top$, $t_{11} \neq 0$, $t_{21}\neq 0$, and $t_{11}\neq t_{21}$.  We can prove the following proposition.

\begin{proposition}\label{prop:k3q1}
Consider model \eqref{modelk3q1} and let $r_{\alpha_1, \alpha_2}$ be the following polynomial of degree 4:
\begin{align*}
   r_{\alpha_1, \alpha_2}(x)=\alpha_1 (-1 + 7 \alpha_1 - 12 \alpha_1^2 + 6 \alpha_1^3) x^4 + 4 \alpha_1 \alpha_2(1 - 6 \alpha_1 + 6 \alpha_1^2) x^3 \\ +
 6 \alpha_1 \alpha_2 (1 - 2 \alpha_2 + \alpha_1 (-2 + 6 \alpha_2)) x^2 + 
 4 \alpha_1 \alpha_2 (1 - 6 \alpha_2 + 6 \alpha_2^2) x + \alpha_2 (-1 + 7 \alpha_2 - 12 \alpha_2^2 + 6 \alpha_2^3).
\end{align*}
All eigenvalues of ICS for the scatter pair $\cov-\cov_4$ are equal to one if and only if the first coordinates $t_{11}$ and $t_{21}$ of the group centers 
are such that 
\begin{equation}\label{eq:req}
r_{\alpha_1, \alpha_2}(t_{11}/t_{21})=0.
\end{equation}
\end{proposition}
The proof is given in Appendix C.
Proposition \ref{prop:k3q1} demonstrates that the behavior of the ICS eigenvalues for a mixture of three Gaussian groups with aligned group centers differs from the two-group case (see Subsection \ref{subsec:2gr}). For model \eqref{modelk3q1}, the condition under which all eigenvalues of $\cov^{-1}\cov_4$ equal one depends not only on the group proportions but also on the group centers, as determined by the fourth-order \eqref{eq:req}.
This equation has four solutions in the complex plane but, between zero and four solutions in the real space, depending on the group proportions. Using Mathematica \citep{inc_mathematica_nodate} to simplify the computations, we found that for $\alpha_1=\alpha_2\leq(3-\sqrt{3})/12$ ($\simeq 10.6\%$), or $\alpha_1=\alpha_2$ equal to $1/3$ or $1/4$, \eqref{eq:req} has no real solution. This implies that one and only one eigenvalue of ICS differs from one, and is associated with an eigenvector that spans the FDS (which is only one-dimensional for model \eqref{modelk3q1}). However, we can also identify group proportions for which, given specific group centers, all eigenvalues of ICS equal one, indicating that ICS does not work in all situations. For instance, when $\alpha_1=\alpha_2=1/6$, \eqref{eq:req} has three real solutions, and all the eigenvalues of ICS equal one if $t_{21}/t_{11}$ equal -1, $(2\sqrt{6}+7)/5$ or $(-2\sqrt{6}+7)/5$. While these situations are highly specific, they illustrate that analysing the eigenvalues of $\cov^{-1}\cov_4$ becomes more complex when $q$ is less than $k-1$, and that this phenomenon is already true for $k=3$.

\section{Empirical study}\label{sec:emp_study}
In Section~\ref{sec:iraccen}, we derive a theoretical understanding of the behavior of the eigenvalues of $\bo{V}_1^{-1}\bo{V}_2$ in the context of a mixture of Dirac distributions (model~\eqref{modeldc}). Firstly, for a mixture model with full rank group centers matrix and in the absence of within-group variability, the eigenvalues of $\bo{V}_1^{-1}\bo{V}_2$, depend only on the group proportions (see Proposition~\ref{prop:eigen}). Secondly, in Subsection~\ref{subsec:dpropk}, the thresholds for the proportion of groups that result in an eigenvalue transitioning from less than the eigenvalue of multiplicity $p-q$ to greater than it, have been identified for different number of groups. Both analyses are restricted to the case where the dimension of the FDS is $q = k-1$. 
To confirm those results in a more general context but still with $q=k-1$, we perform simulations of a mixture of Gaussian distributions, including within-group variability and noise, and we also focus on different scatter pairs.
Subsection~\ref{subsec:simus_design} describes the simulations settings, Subsection~\ref{subsec:simus_means} studies the eigenvalues of  $\cov^{-1}\cov_4$ when group centers vary in the presence of within-group variability, while Subsection~\ref{subsec:simus_scatters} focuses on the behavior of eigenvalues for different scatter pairs. 

\subsection{Simulation design\label{subsec:simus_design}}
We generate $n=1000$ observations from a particular case of model~\eqref{modelt} of a mixture of Gaussian distributions with $k$ different group means and equal covariance matrix, under the assumption that the group centers matrix is of full rank $q = k - 1$ as in \cite{alfons_tandem_2024}:
\begin{equation*}\label{mixt}
\bo{X} \sim \sum_{j=1}^{k} \alpha_j \, {\cal N}(\bo{t}_j,\bo{I}_p),
\end{equation*}
where $\alpha_j>0$ are mixture weights such that $\sum_{j=1}^k\alpha_j=1$, for $j\in\{1,\ldots, k\}$, $\bo{t}_1 = 0$, and $\bo{t}_{l+1} = \delta \bo{e}_l$, for $l \in \{1, \ldots, k-1\}$,  $\bo{e}_l$ is a $p$-dimensional vector with one in the $l$-th coordinate and zero elsewhere, for different number of clusters $k \in \{2,3,5,10\}$,  on $p=5k$ variables. In this setting, the cluster structure lies in a low-dimensional subspace of dimension $q=k-1$. We consider 17 scenarios of group proportions ($\alpha_j$, $j\in\{1,\ldots, k\}$): 
\begin{itemize}
    \item $k=2$ clusters: 50--50, 40--60, 30--70, 21--79, 10--90,
    \item $k=3$ clusters: 33--33--34, 18--35--50, 10--40--50, 10--30--60, 10--80--10,
    \item $k=5$ clusters: 20--20--20--20--20, 14--20--20--20--26, 10--20--20--20--30,
    10--10--20--20--40,
    \item $k=10$ clusters: 10--10--10--10--10--10--10--10--10--10,
                    8--10--10--10--10--10--10--10--10--12,
                    5--5--5--5--5--5--5--15--20--30.
\end{itemize}

We consider  $\delta \in \{1,5,10,50,100\}$ for the study on different group centers in Subsection~\ref{subsec:simus_means}. For the study on different scatters in Subsection~\ref{subsec:simus_scatters}, we fix $\delta=10$ and we focus on the following scatter pairs:
$\cov-\cov_4$,
$\covAxis-\cov$,
$\tcov-\cov$ and
$\mcd_{\tau}-\cov$   with $\tau \in \{0.25, 0.5, 0.75\}$.  
$\covAxis$ is a one-step M-estimator using the inverse weight function of the squared Mahalanobis distance, $\tcov$ is also a
 one-step M-estimator but with weights based on pairwise Mahalanobis distances and $\mcd_{\tau}$ are the (raw) minimum covariance determinant estimators based on $n_\tau = \lceil \tau n \rceil$ observations for which the sample covariance matrix has the smallest determinant. 
 See Subsection~3.1 of \cite{alfons_tandem_2024} for details on those scatter matrices. Note that we follow \cite{alfons_tandem_2024} and take $\bo{V}_1$ being more robust than $\bo{V}_2$. In addition, $\mcd_{\tau}$ are computed using the FAST-MCD algorithm of \citep{rousseeuw_fast_1999}. For the selection of components, we choose the med criterion as introduced in \cite{alfons_tandem_2024}, which selects the $k-1$ components whose eigenvalues deviate the most from the median of all eigenvalues. This test relies on the assumption that the dimension of interest $q$ is low compared to the number of variables $p$: $q \leq p/2$, which implies that the majority of the eigenvalues should be equal, which is met in our context.   We simulate 50 data sets for each of the different settings. 

\subsection{\texorpdfstring{$\cov^{-1}\cov_4$}{} eigenvalues when group centers vary in the presence of within-group variability and 
\texorpdfstring{$q=k-1$}{}\label{subsec:simus_means}}

As proven in Proposition~\ref{prop:eigen}, theoretically, for a mixture model with full rank group centers matrix and in the absence of within-group variability, the eigenvalues of $\bo{V}_1^{-1}\bo{V}_2$ depend only on the group proportions $\alpha_j$, $j\in\{1,\ldots, k\}$ of the mixture components. Here, instead of computing the eigenvalues theoretically, we simulate a mixture of Gaussian distributions with some within-group variability to evaluate the sensitivity of this proposition and we compare with the results in Fig.~\ref{fig:boxplot_threshold} for $\cov^{-1}\cov_4$. In Fig.~\ref{fig:boxplot_eigenvalues_simus}, we display the boxplots of the $k-1$ first and the $k-1$ last eigenvalues of $\cov^{-1}\cov_4$ over 50 replications when the group centers vary in the presence of within-group variability. Grey-shaded areas expose the eigenvalues which are theoretically different from one. 
The cases for which the eigenvalue one has multiplicity greater than $p-q$ are identified by the red-shaded areas.

\begin{figure}[!ht]
\begin{center}
\includegraphics[width=1\textwidth]{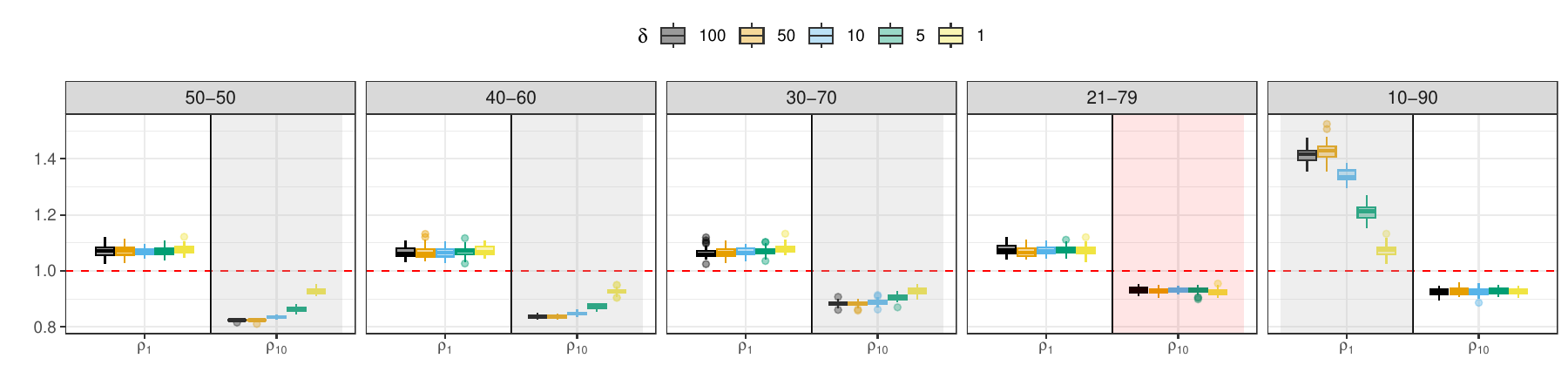}
\includegraphics[width=1\textwidth]
{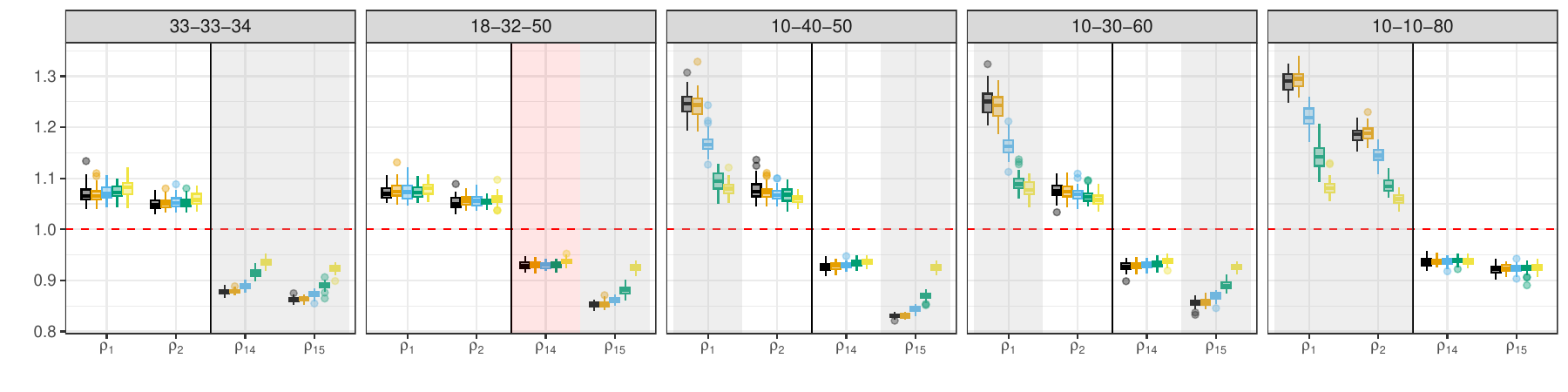}
\includegraphics[width=1\textwidth]
{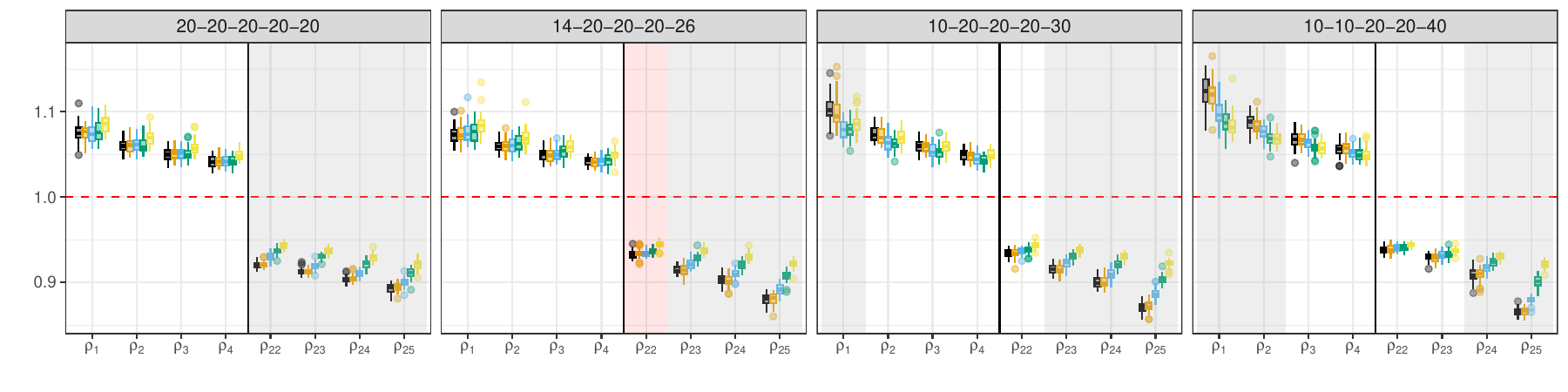}
\includegraphics[width=1\textwidth]
{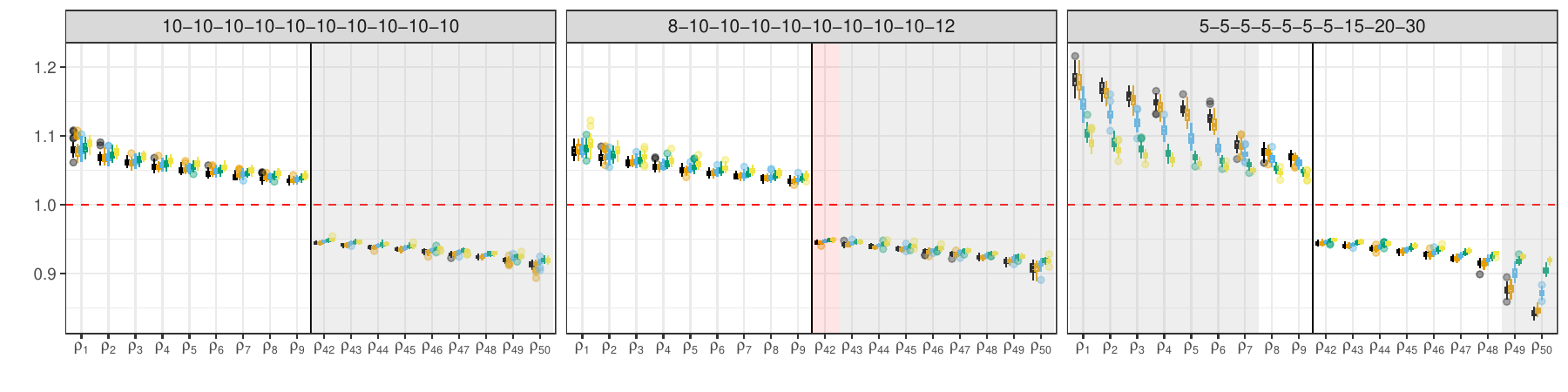}
\caption{Boxplots of the $k-1$ first and the $k-1$ last eigenvalues of $\cov^{-1}\cov_4$ when the group centers vary in the presence of within-group variability, for $q = k-1$, $p=5k$, $n=1000$, and different values of $k$ over 50 replications. Results for different group proportions are shown in different panels. Grey-shaded areas expose the eigenvalues which are theoretically different from one. Red-shaded areas highlight the cases for which the eigenvalue one has multiplicity greater than $p-q$.}
\label{fig:boxplot_eigenvalues_simus}
\end{center}
\end{figure}

Contrary to Fig.~\ref{fig:boxplot_threshold}, we can see some variability for the eigenvalues over the replications and between the different group centers, determined by the values of $\delta$. In addition, given the presence of within-group variability and noise with $p>q=k-1$, none of the eigenvalues are strictly equal to one. However, the eigenvalues of interest for highlighting the groups' structure (in grey-shaded areas) are still easily identifiable as the ones further away from one. For example, for the first row with scenarios of a mixture of two components, only the first or the last eigenvalue should be different from one. More specifically, as illustrated in Subsection~\ref{subsec:dpropk}, for a mixture with group proportion of $50-50$, only the last eigenvalue is different from one, as we can see in the first subplot. If one group proportion decreases until being below 21\% then it is now the first eigenvalue which is different from one. For this threshold of 21\%, all eigenvalues are equal to one as highlighted by the red-shaded area. We can identify the same behavior when the number of groups increases. For example, in the second row, if all the three group proportions are ``large enough'' then the last two eigenvalues are different from one. When only one group proportion is ``low'', as in the scenario $10-40-50$, the first and the last components are different from one. When two group proportions are ``low'', like in the scenario $10-10-80$, then only the first two components are different from one. This pattern repeats itself for mixtures of 5 or 10 groups as illustrated in the third and fourth rows. The thresholds are clearly identifiable and equal to 18\% for 3 groups, 14\% for 5 groups and 8\% for 10 groups, confirming those identified in Subsection~\ref{subsec:dpropk} in the absence of within-group variability. 

It is also important to note that the variability of the eigenvalues between group centers, i.e., for different values of $\delta$, is smaller for the eigenvalues which are theoretically equal to one (in white or red-shaded areas) no matter the value of $\delta$. This supports evidence for some stability of the eigenvalues associated with no structure of the data, no matter the group centers.  Clearly, at the thresholds, identified by red-shaded areas, there is almost no variability of the eigenvalues, as we can see in the first row, in the well-known case of a mixture of two components with group proportions $21-79$. 
For the eigenvalues of interest (in grey-shaded areas), some variability is visible but they seem to quickly converge as $\delta$ increases. In addition, the eigenvalues are clearly different from one as soon as $\delta \ge 5$. For $\delta=1$, the ratio signal over noise is too low and so, all eigenvalues are close to one. On the contrary, if $\delta \ge 50$, the ratio signal over noise is high and the noise is almost null so, we are almost in the same context as for the Proposition~\ref{prop:eigen}, and the eigenvalues with $\delta = 50$ or $\delta = 100$ are almost equal. In fact, it appears that the eigenvalues converge more slowly in the presence of a low group proportion, as observed in scenarios where at least one group has a proportion of 5\% or 10\%. In this context, the eigenvalues are also higher in absolute value for any $\delta$ so it is one of the favorable cases for ICS since it is easy to identify it is different from one. 
So, for the next study, we consider that the group centers have a negligible effect, and we focus on the case $\delta = 10$.

\subsection{\texorpdfstring{$\bo{V}_1^{-1}\bo{V}_2$}{} eigenvalues for different scatter pairs in the presence of within-group variability and for 
\texorpdfstring{$q=k-1$}{} \label{subsec:simus_scatters}}

In this section, we extend our study to the behavior of the eigenvalues of $\bo{V}_1^{-1}\bo{V}_2$ for different affine equivariant scatter pairs and not only to $\cov-\cov_4$. In such a context, and in the absence of within-group variability, only the group proportions should impact the eigenvalues. Here, we introduce some within-group variability and we analyze if the general behavior for the eigenvalues of $\cov^{-1}\cov_4$ is the same for different scatter pairs. More specifically, we want to know if the thresholds are linked to the same group proportions and if only ``low'' group proportions are associated with first eigenvalues being different from each other. Deriving theoretical eigenvalues for more complex scatter pairs is a difficult task so, we focus on the analysis of the invariant components instead. Indeed, to highlight the data structure, we have to project the data onto the subspace spanned by the eigenvectors associated with eigenvalues which are not equal to each other. So we can hypothesize that the invariant components to select are the ones associated with the eigenvalues which are not equal to each other. To identify the components of interest, we followed the recommendations of \cite{alfons_tandem_2024} and we chose the $k-1$ invariant components based on the med criterion, i.e., the ones associated with eigenvalues which deviate the most from the median of all eigenvalues.

In Fig.~\ref{fig:heatmaps_IC_simus}, we display heatmaps of the percentage of selections for the $k-1$ first and the $k-1$ last invariant components of $\bo{V}^{-1}\bo{V}_2$ by the med criterion, for different number of groups, for group proportions and for different scatter pairs. If the value is 100\%, it indicates that the eigenvalue associated with this component is stable and different from the others. If the percentage is less than 100\%, it means that, over the replications, the criterion has not consistently chosen the same eigenvalue, suggesting that the eigenvalue might not be significantly different from the others.
For example, in the first subplot, for the scatter pair $\cov-\cov_4$, the value is 100\% for the case of a mixture of two Gaussian distributions with group proportions of 10-90.  We can deduce that the first component IC$_1$ is always chosen over the 50 replications. This is in line with the previous conclusion which states that if one group proportion is ``low" then the first eigenvalue is different from one. In case of 50-50, it is always the last one which is chosen and of interest. For group proportions in between, the choice is not as clear mainly for the scenario 21-79. In this case, IC$_1$ is selected instead of IC$_{10}$ only in 75\% of cases. When the selection is not clear, this is a sign that the eigenvalues are not as far away from the others and it might indicate a threshold. Precisely, this result confirms the theory recalled in Subsection~\ref{subsec:2gr} that, when a group proportion is equal to 21\% and in presence of two groups, then all the eigenvalues are equal to one with $\cov-\cov_4$. In practice, the group proportion associated with the threshold is not as precise as in the theoretical situation since, even when a group proportion is around 30\%, it can be difficult to have a clear separation between eigenvalues in 15\% of cases.

\begin{figure}[!ht]
\begin{center}
\includegraphics[width=1\textwidth]{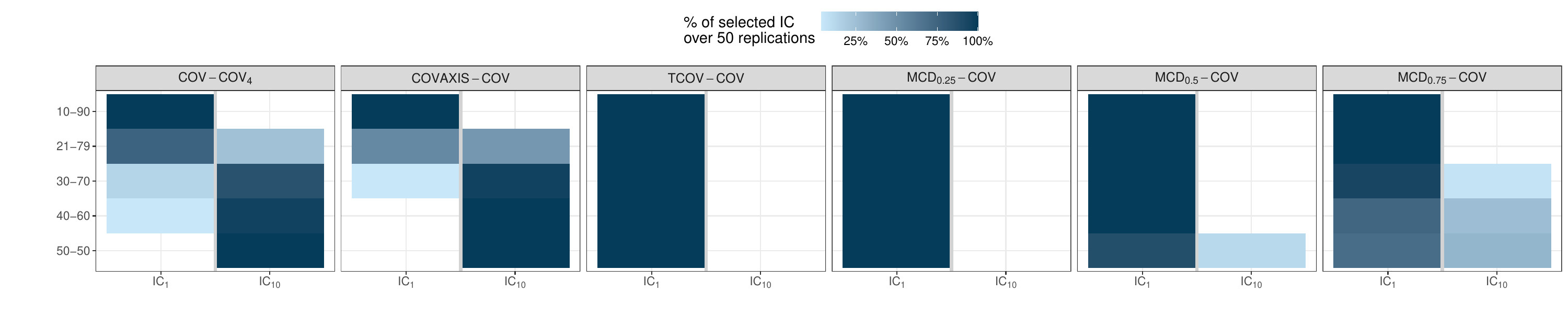}
\includegraphics[width=1\textwidth]{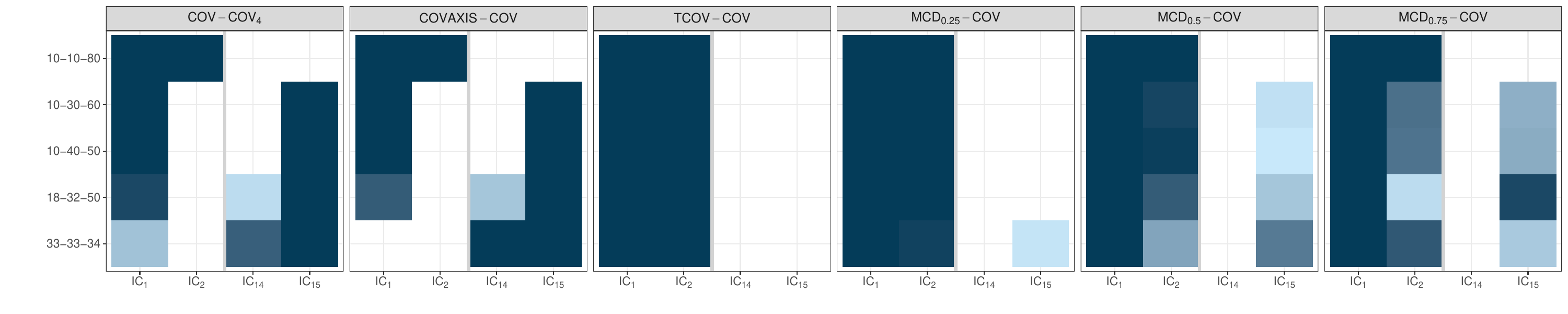}
\includegraphics[width=1\textwidth]{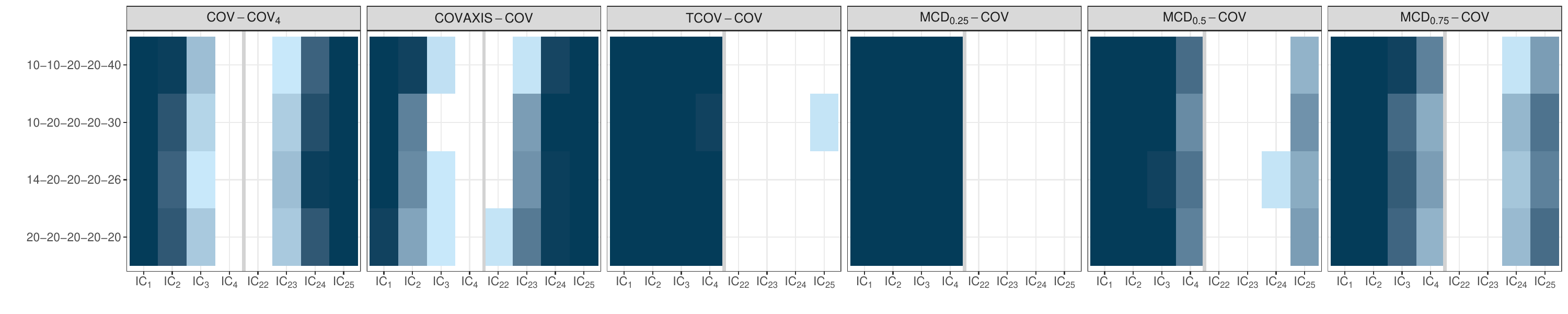}
\includegraphics[width=1\textwidth]{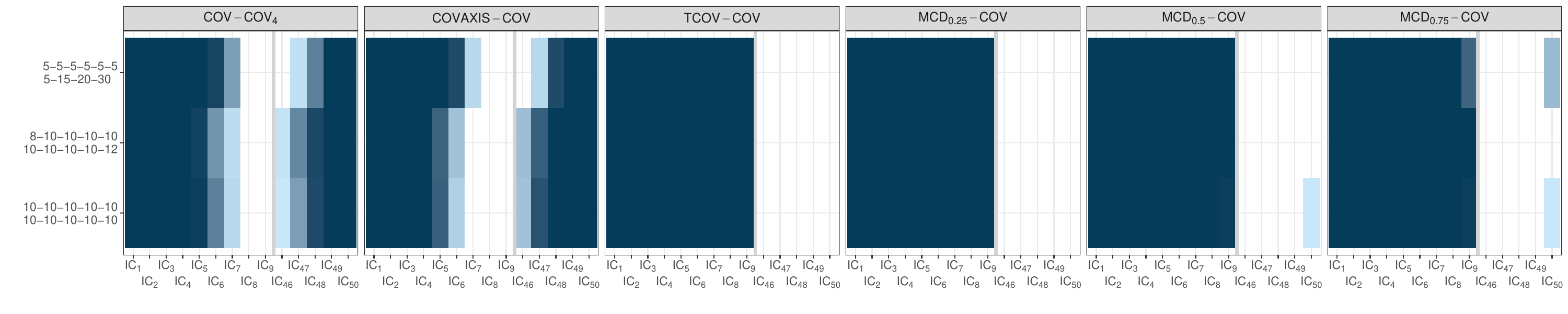}
\caption{Heatmaps of the percentage of selection for the $k-1$ first and $k-1$ last selected invariant components of $\bo{V}^{-1}\bo{V}_2$ by the med criterion, in the presence of within-group variability, for $q = k-1$, $p=k/0.2$, $n=1000$, different values of $k$ over 50 replications. Results for different number of groups are shown in separate rows,  for different group proportions in the y-axis and for different scatter pairs in different panels.}
\label{fig:heatmaps_IC_simus}
\end{center}
\end{figure}

 Overall, for $\cov-\cov_4$ we can corroborate the same thresholds as mentioned in the previous section: approximately 21\% for 2 groups, 18\% for 3 groups, 14\% for 5 groups and 8\% for 10 groups. However, because we infer the behavior of the eigenvalues based on the selection of components and we are in the context of within-group variability, noise and simulated data, the values of the thresholds are not as precise as in the theoretical cases illustrated in Fig.~\ref{fig:boxplot_threshold}.
 Nevertheless, this analysis globally confirms the results already observed and demonstrated in Subsection~\ref{subsec:dpropk}. It is also quite easy to visually identify on which components the structure of interest relies depending on the group proportions and to extend our understanding of such behavior for the other scatter pairs. 

 The scatter pair $\covAxis-\cov$ is almost exhibiting the same behavior as $\cov-\cov_4$ as well as for the thresholds. This is not surprising since $\covAxis$ is also a one-step M-estimator as $\cov_4$ but uses an inverse weight function of the squared Mahalanobis distance instead. For $\tcov-\cov$, the choice of the components seems to be easy and always associated with the first eigenvalues, except for one case with 5 groups (10-20-20-20-30). For the scatter pairs based on $\mcd_{\tau}$, the patterns are conditional on the value $\tau$. It is interesting to note that the behavior for $\mcd_{0.25}-\cov$ looks almost the same as for $\tcov-\cov$, apart from one case with three balanced groups. For $\mcd_{0.50}-\cov$ then the situation is no longer perfect and sometimes the last eigenvalue is chosen, e.g., in the scenario of group proportions of 50-50. Finally, for $\mcd_{0.75}-\cov$ the situation looks worse and the choice is not clear in a lot of cases, even more than for the scatter pair $\cov-\cov4$. These results support the recommendations of using $\tcov-\cov$ or $\mcd_{0.25}-\cov$ when clustering is the goal \cite{alfons_tandem_2024}.

 To conclude, it appears that those two scatter pairs, $\tcov-\cov$ or $\mcd_{0.25}-\cov$, seem to behave in a simpler manner than $\cov-\cov_4$ by finding groups only on the first components no matter the group proportions and with fewer thresholds. In addition, the more groups are present, the easier it looks to find them.  However, we have shown that in a few scenarios, those scatter pairs might fail to identify some groups. Since our analysis is not exhaustive on the number of scenarios, we can suppose it might happen in other situations. One idea to overcome such a potential issue, without knowing the theoretical eigenvalues, is to run ICS multiple times, with different scatter pairs and to combine and analyze the selected components from multiple pairs to be sure to detect the entire group structure. 
 Another possibility is to perform localized PP after ICS to see which directions are interesting as suggested in \cite{dumbgen_refining_2023}.

\section{Empirical application\label{sec:application}}

To extend our findings beyond theoretical distributions, we analyze the Olympic decathlon dataset from the 2016 Rio de Janeiro Games \citep{loperfido_maxskew_2017}, which contains scores across 10 events for 23 athletes. While PCA only identifies two outlying athletes (Saluri and Nakamura) in its first two components, \cite{loperfido_skewness-based_2018} showed that skewness-based projection pursuit reveals additional structure. Notably, it detects the athleteTaiwo, who despite ranking 11$^{th}$ overall, displayed unusual performance patterns – achieving both the highest High Jump score and nearly the lowest Javelin score.


As shown in Fig.~\ref{fig:application} (first row), ICS with $\cov-\cov_4$ identifies Saluri, Taiwo, and Nakamura on IC$_1$, IC$_2$, and IC$_3$ respectively, complementing both PP and PCA findings by detecting all three outliers. This matches our theoretical results from Sections~\ref{sec:iraccen}, \ref{sec:dprop}, and \ref{sec:emp_study}, where we showed that eigenvalue transitions depend solely on group proportions. For $\cov-\cov_4$, the threshold decreases from 21\% (two groups) to 8\% (ten groups) in mixture distributions. Viewing the data as a mixture model with three outliers (Taiwo, Nakamura, and Saluri, each ~4.3\%) and a majority group (86.9\%) explains why these athletes are detected on the first invariant components, consistent with common mixture modeling of outliers (\cite{archimbaud_ics_2018}, \cite{hou_re-centered_2014}).

To further analyze eigenvalue behavior, let us add a new athlete to the dataset who behaves similarly to Saluri, the lowest-ranked athlete. We apply a uniform $U_{[-20,20]}$ noise to each of Saluri's scores. This yields mixture weights of approximately $\alpha_{Saluri} \approx 8.3\%$, $\alpha_{Taiwo} = \alpha_{Nakumara} \approx 4.1\%$, and $\alpha_{majority} \approx 83.3\%$. As shown in the second row of Fig.~\ref{fig:application}, Saluri is no longer an outlier, while Taiwo appears on the first component and Nakamura on the second.
Next, we expand the Saluri-based cluster to weights of $\alpha_{Saluri} \approx 12\%$, $\alpha_{Taiwo} = \alpha_{Nakumara} \approx 4\%$, and $\alpha_{majority} \approx 80\%$ by adding another synthetic observation. In the last row of Fig.~\ref{fig:application}, Taiwo and Nakamura still appear as outliers on the first two components, but the Saluri-based cluster now stands out on the last component. This supports the theory that at a given threshold, one eigenvalue transitions from being smaller to larger than the eigenvalue of multiplicity $p - q$. For this dataset, the transition happens when a group reaches roughly 8\%: below this, it appears on the first component, and above this, on the last. Notably, this 8\% threshold aligns with an assumption of 10 groups rather than 4 (see Section \ref{sec:dprop}), suggesting that the data may naturally form 10 clusters, matching its 10 variables. Thus, even with unknown data distribution, the theoretical results hold and may provide insights into the number of clusters.

\begin{figure}[!ht]
\begin{center}
\includegraphics[width=1\textwidth]{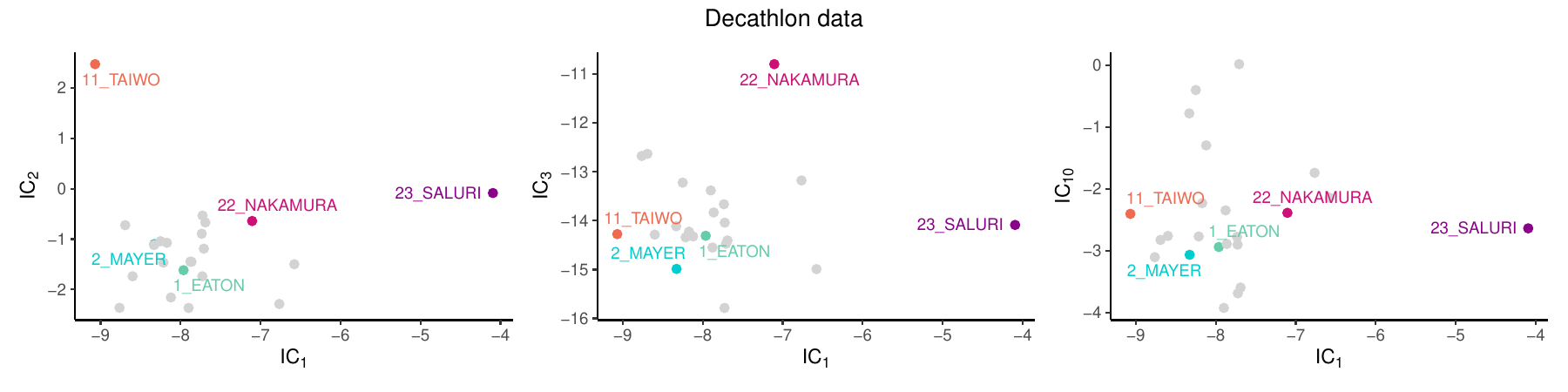}
\includegraphics[width=1\textwidth]{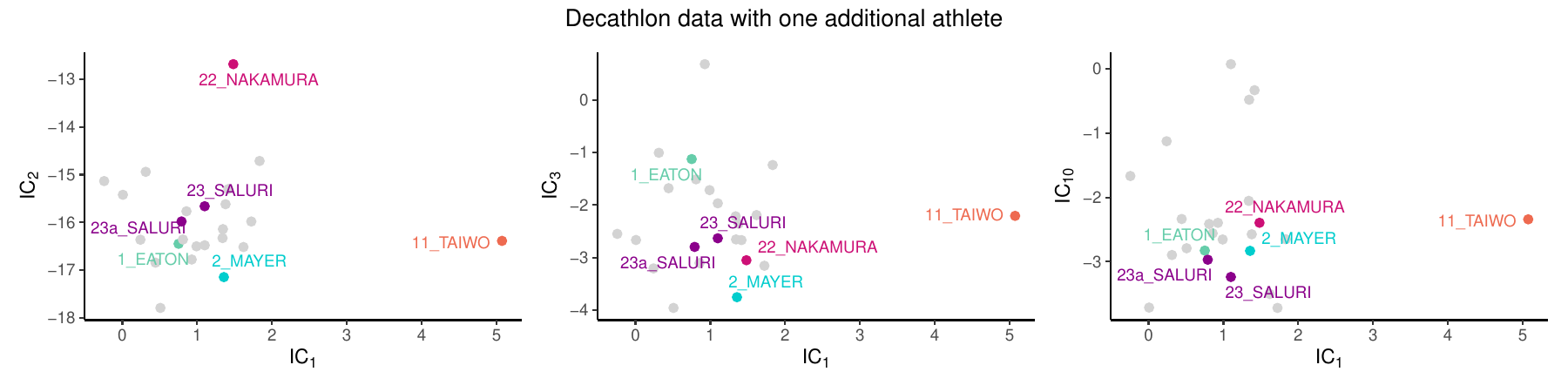}
\includegraphics[width=1\textwidth]{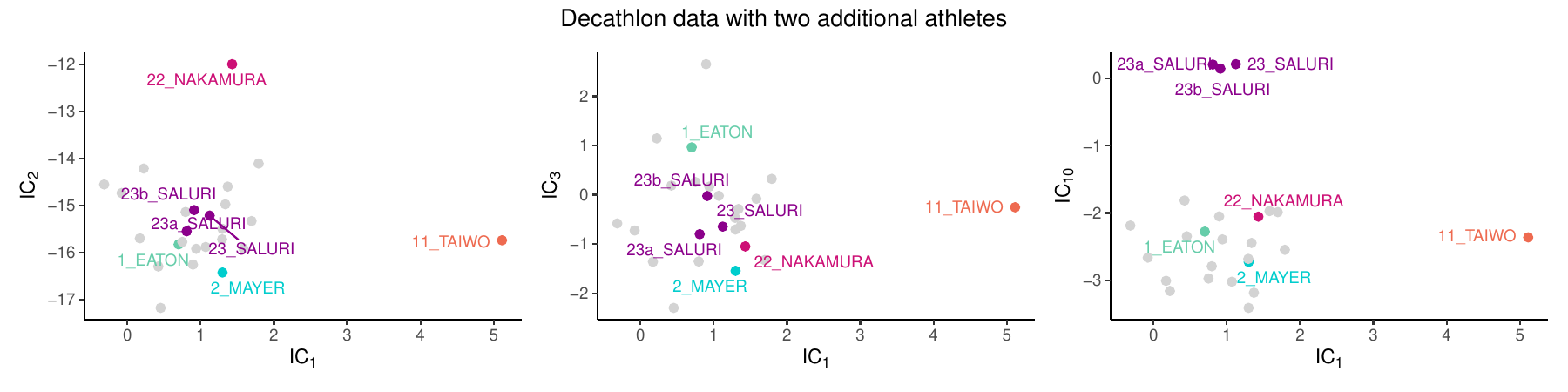}

\caption{Scatterplots of the second, third and last ICS components versus the first component with $\cov-\cov_4$ on the decathlon data set. Results for different modified data sets are shown in separate rows. The names of the athlete are preceded by their rank. The colors identify the athletes coming from the same cluster.}
\label{fig:application}
\end{center}
\end{figure}

\section{Conclusions and perspectives\label{sec:conclusion}}

Dimension reduction is becoming increasingly important. PCA is probably the most utilized method in practice due to its simplicity, despite lacking guarantees for its effectiveness as a preprocessing tool for clustering or outlier detection. In these contexts, PP appears much more natural and has theoretical justification \cite{radojicic_large-sample_2021}. However, PP is computationally expensive. From that perspective, ICS is a promising alternative -- it is computationally less demanding, and theoretical properties can be derived in quite general mixture model frameworks. As shown in the seminal paper \cite{tyler_invariant_2009}, ICS can recover the FDS. It essentially reduces to an eigenvalue problem where the noise space has identical eigenvalues. The crucial question is then whether all eigenvalues belonging to the space spanning the FDS are distinct from the noise value. In the two-group Gaussian mixture model using the combination $\cov - \cov_4$, it was known that this generally works except for one specific mixing proportion. In this work, we extended this result to cover richer mixture models and different scatter combinations, though theoretical studies seem mainly feasible only for $\cov - \cov_4$. Based on the findings, it seems that ICS is indeed a natural dimension reduction method for clustering and outlier detection, and with an increasing number of clusters ($k$ still smaller than $p$), the FDS will be estimated.

Based on the current paper, it would be worthwhile to pursue extending these results to even richer mixture models and also consider in more detail other scatter combinations. It could also be investigated whether, in cases where one scatter combination fails, there are other combinations that will work, or if there exists a global worst-case scenario.

In practice, it seems customary to compare the performance of different scatter combinations and choose the best one, which is, however, still done heuristically, and corresponding tools for a comparison could be developed. One crucial issue in practice is then also to establish what the noise space eigenvalue is, and which eigenvalues are distinct from that value. Some heuristic rules are, for example, discussed in \cite{archimbaud_ics_2018,alfons_tandem_2024,radojicic_non-gaussian_2020}, but inferential tools are still missing. So far, only \cite{kankainen_tests_2007} propose some tests using $\cov - \cov_4$ if all eigenvalues are equal in the Gaussian case (i.e., testing for multivariate normality) and \cite{luo_combining_2016, radojicic_non-gaussian_2020, nordhausen_asymptotic_2022} in a non-Gaussian component analysis framework for the equality of eigenvalues for components belonging to Gaussian components. Similar tests might be of interest also in model \eqref{model0}.



\section*{Computational details}
All computations in Sections~\ref{sec:icsFDS}, \ref{sec:iraccen}, \ref{sec:dprop} are performed with Python version 3.10.13 \citep{python3}, 
notable packages include \texttt{NumPy} \citep{harris2020array} for numerical computations, \texttt{Pandas} \citep{reback2020pandas, mckinney-proc-scipy-2010} for data manipulation and \texttt{Plotly} \citep{plotly} for data visualization in Subsection~\ref{subsec:cc4calc}, Sections~\ref{sec:iraccen} and \ref{sec:dprop}. 
Furthermore, all simulations in Section~\ref{sec:emp_study} and the analysis of the application in Section~\ref{sec:application} are performed with \proglang{R} version 4.3.3 \citep{r_core_team_r_2023} and uses the \textsf{R} packages \pkg{ICS} \citep{nordhausen_tools_2008, nordhausen_ics_2023} for ICS, \pkg{ICSClust} \citep{archimbaud_icsclust_2023} for selecting the components, \pkg{rrcov} \citep{todorov_object-oriented_2009} for the MCD scatter matrix.  
Replication files for the theoretical computations and simulations are available from \url{https://github.com/AuroreAA/ICS_FDS-Replication}.
\section*{Acknowledgments}

We thank the editor, associate editor and referees for their constructive remarks.
We thank Camille Mondon for stimulating discussions about ICS.
Part of Colombe Becquart's work was done during an Erasmus+ visit funded by the University of Toulouse at the University of Helsinki.
Anne Ruiz-Gazen acknowledges funding from the French National Research Agency (ANR) under the Investments for the Future (Investissements d’Avenir) program, grant ANR-17-EURE-0010. Klaus Nordhausen was supported by the HiTEc COST Action (CA21163) and the Research Council of Finland (363261).

\section*{Appendix A. \, Calculation details for the scatter pair \texorpdfstring{$\cov-\cov_4$}{}}

\subsection*{Appendix A.1. \, The case of Gaussian mixture in Subsection~\ref{subsec:cc4calc}}
Because of the invariance property of ICS, we consider w.l.o.g. \citep[see A.2 in][]{tyler_invariant_2009} the mixture~:
$ \displaystyle \bo{X} \sim \sum_{j=1}^{k}\alpha_j {\cal N}_p(\bo{t}_j, \bo{I}_p) $,
where the $\bo{t}_j = (t_{ji}) \in \mathbb{R}^p$ are distinct for $j \in \{1,\ldots, k\}$ and can be written $ \sum_{i=1}^{j} t_{ji} \bo{e}_i$ (with \(\bo{e}_{i} \) the $p$-dimensional vector with one in the $i$-th coordinate and zero elsewhere for $j\in\{1,\ldots, k-1\}$), and $\bo{t}_k$ is the  $p$-dimensional zero vector. Furthermore, to ease our computations, we define $\bo{X}^{\mbox{c}}:=\bo{X}-\E(\bo{X})$ whose distribution is $\displaystyle  \sum_{j=1}^{k}\alpha_j {\cal N}_p(\bo{t}_j^{\mbox{c}}, \bo{I}_p)$, 
where $\bo{t}_j^{\mbox{c}} = \bo{t}_j - \E(\bo{X})$ has coordinates $t_{ji}^{\mbox{c}}$, for $j\in\{1,\ldots, k\}$ and $i\in\{1,\ldots, q\}$. The covariance matrix of $\bo{X}$ is
$\cov= \bs{\Gamma}_W + \bs{\Gamma}_B$
where $\bs{\Gamma}_{W}=\bo{I}_p$ is the within-group covariance matrix and $\bs{\Gamma}_{B}=\sum_{j=1}^{k}\alpha_j \bo{t}_{j}^{\mbox{c}} (\bo{t}_{j}^{\mbox{c}})^\top$ is the between-group covariance matrix. 
This yields:
\begin{align*}
\cov= \begin{bmatrix} \bs{\beta} & \bo{0} \\ \bo{0} & \bo{I}_{p-q} \end{bmatrix}, \quad \cov^{-1} = \begin{bmatrix} \bo{B} & \bo{0} \\ \bo{0} & \bo{I}_{p-q} \end{bmatrix},
\end{align*}
where the terms of $\bs{\beta}$ are $\beta_{ms} = \sum_{j=1}^{k} \alpha_j t_{jm}^{\mbox{c}}t_{js}^{\mbox{c}}$ for distinct $m, s \in\{1,\ldots, q\}$, and $\beta_{mm} = 1 + \sum_{j=1}^{k} \alpha_j (t_{jm}^{\mbox{c}})^2$ for $m\in~\{1,\ldots, q\}$. $\bo{B} = (b_{ij})$ is $q \times q$ matrix whose terms are difficult to express. Concerning $\cov_4$, we get:
\begin{align*}
\cov_4&=\frac{1}{p+2}\times\E\big((\bo{X}^{\mbox{c}})^\top \cov^{-1}\bo{X}^{\mbox{c}}\bo{X}^{\mbox{c}}(\bo{X}^{\mbox{c}})^\top\big) 
 = \frac{1}{p+2}\times\E\Bigg[ \Bigg(\sum_{i=1}^{q}\sum_{j=1}^{q}x^{\mbox{c}}_ix^{\mbox{c}}_jb_{ij} +\sum_{i=q+1}^{p}(x^{\mbox{c}}_i)^2 \Bigg) \bo{X}^{\mbox{c}}(\bo{X}^{\mbox{c}})^\top\Bigg],
\end{align*} 
where $x^{\mbox{c}}_i$ are the coordinates of $\bo{X}^{\mbox{c}}$, for $i\in\{1,\ldots, p\}$. We proceed to compute the moments of the coordinates of $\bo{X}^{\mbox{c}}$ to obtain the final form of $\cov_4$. 

Let $f_{X}(\mu) = \mu$, $f_{X^2}(\mu, \sigma) = \mu^2+\sigma^2$, $f_{X^3}(\mu, \sigma) = \mu^3+3\mu\sigma^2$, and $f_{X^4}(\mu, \sigma) = \mu^4+6\mu^2\sigma^2+3\sigma^4$ be the functions that give the moments of order 1 to 4 of  the univariate normal distribution with mean $\mu$ and variance $\sigma^2$. 
The moments of the mixture are equal to the linear combination of the moments of its components. For $i,j \in\{1,\ldots, q\}$, $x^{\mbox{c}}_i$ and $x^{\mbox{c}}_j$ are not independent, but  for each mixture component, the coordinates are independent and we get the following expressions: 
\begin{equation*}
\begin{aligned}
    \E[x^{\mbox{c}^2}_a] &= \sum_{j=1}^{k} \alpha_j f_{X^2}(t^{\mbox{c}}_{ja}, 1), \quad
    \E[x^{\mbox{c}}_a x^{\mbox{c}}_b] = \sum_{j=1}^{k} \alpha_j f_{X}(t^{\mbox{c}}_{ja}) f_{X}(t^{\mbox{c}}_{jb}), \quad
    \E[x^{\mbox{c}^4}_a] = \sum_{j=1}^{k} \alpha_j f_{X^4}(t^{\mbox{c}}_{ja}, 1), \\
    \E[x^{\mbox{c}^3}_a x^{\mbox{c}}_b] &= \sum_{j=1}^{k} \alpha_j f_{X^3}(t^{\mbox{c}}_{ja}, 1) f_{X}(t^{\mbox{c}}_{jb}), \quad 
    \E[x^{\mbox{c}^2}_a x^{\mbox{c}^2}_b] = \sum_{j=1}^{k} \alpha_j f_{X^2}(t^{\mbox{c}}_{ja}, 1) f_{X^2}(t^{\mbox{c}}_{jb}, 1),\\
    \E[x^{\mbox{c}^2}_a x^{\mbox{c}}_b x^{\mbox{c}}_c] &= \sum_{j=1}^{k} \alpha_j f_{X^2}(t^{\mbox{c}}_{ja}, 1) f_{X}(t^{\mbox{c}}_{jb}) f_{X}(t^{\mbox{c}}_{jc}), \quad
    \E[x^{\mbox{c}}_a x^{\mbox{c}}_b x^{\mbox{c}}_c x^{\mbox{c}}_d] = \sum_{j=1}^{k} \alpha_j f_{X}(t^{\mbox{c}}_{ja}) f_{X}(t^{\mbox{c}}_{jb}) f_{X}(t^{\mbox{c}}_{jc}) f_{X}(t^{\mbox{c}}_{jd}),
\end{aligned},
\end{equation*}
for distinct $a, b, c, d \in\{1,\ldots, q\}$.
For $i > q$ and $j\in\{1,\ldots, p\}$, $x^{\mbox{c}}_i$ and $x^{\mbox{c}}_j$ are independent which leads to many elements of $\cov_4$ being equal to $0$ and to the following expression: 
$\cov_4= \begin{bmatrix} \bs{\Psi} & \bo{0} \\ \bo{0} & \bo{I}_{p-q} \end{bmatrix}$
where $\bs{\Psi} = (\psi_{ms})$ is $q \times q$ matrix whose terms for $m,s \in\{1,\ldots, q\}$ are given by:
\begin{align*}
    \psi_{ms} &= \frac{1}{p+2} \E\Bigg[\sum_{i=1}^{q}\sum_{j=1}^{q}x^{\mbox{c}}_m x^{\mbox{c}}_s x^{\mbox{c}}_i x^{\mbox{c}}_j b_{ij} + x^{\mbox{c}}_m x^{\mbox{c}}_s\sum_{i=q+1}^{p}(x^{\mbox{c}}_i)^2\Bigg]
    =\frac{1}{p+2}\Bigg[\sum_{i=1}^{q}\sum_{j=1}^{q} b_{ij}\E[x^{\mbox{c}}_m x^{\mbox{c}}_s x^{\mbox{c}}_i x^{\mbox{c}}_j] +(p-q) \E[x^{\mbox{c}}_m x^{\mbox{c}}_s]\Bigg].
\end{align*}

\subsection*{Appendix A.2. \, The case of Dirac mixture in Section~\ref{sec:iraccen}}

Upon the removal of the noise and the dimensions not associated with FDS, the Gaussian distributions are simplified to Dirac distributions with parameter $\bo{t}_j$, denoted by $\delta_{\bo{t}_j}$: $ \bo{X} \sim \sum_{j=1}^{k}\alpha_j \delta_{\bo{t}_j}, $
where $\bo{t}_{j} = \sum_{i=1}^{j} t_{ji} \bo{e}_{i}$, and \( \bo{e}_{i} \) is a $q$-dimensional vector with one in the $i$-th coordinate and zero elsewhere for $j \in \{1,\ldots, k-1\}$, $\bo{t}_k$ is the zero vector, and the $\bo{t}_j$ are distinct for $j \in \{1,\ldots, k\}$.
The calculations are very similar to the ones in Appendix A.1 with the within-group matrix $\bs{\Gamma}_W$ equals to the zero matrix, which removes some terms in the scatter expressions. In accordance with the previously established notation, it yields $\cov = [\beta_{ms}]$ where 
$\beta_{ms} = \sum_{j=1}^{k}\alpha_jt_{jm}^{\mbox{c}}t_{js}^{\mbox{c}}$
for $m,s \in\{1,\ldots, q\}$, and 

$$\cov_4 = \frac{1}{q+2}\times\E\Bigg[ \Bigg(\sum_{i=1}^{q}\sum_{j=1}^{q}x^{\mbox{c}}_ix^{\mbox{c}}_jb_{ij} \Bigg) \bo{X}^{\mbox{c}}(\bo{X}^{\mbox{c}})^\top\Bigg].$$

The moments of a random variable following a Dirac distribution centered at $t^{\mbox{c}}$ is given by $m_{\nu} = (t^{{\mbox{c}}})^{\nu}$.
All the moments can be expressed as 
$ \E[x^{\mbox{c}}_ax^{\mbox{c}}_bx^{\mbox{c}}_cx^{\mbox{c}}_d] =\sum_{j=1}^{k}\alpha_j t^{\mbox{c}}_{ja} t^{\mbox{c}}_{jb} t^{\mbox{c}}_{jc} t^{\mbox{c}}_{jd} $
for $a, b, c, d \in\{1, \ldots, q\}$. It yields 
$ \displaystyle \cov_4 =[\psi_{ms}]$ where

$$ \displaystyle \psi_{ms} =\frac{1}{q+2} \Bigg[ \sum_{\ell=1}^{k} \alpha_{\ell} t^{\mbox{c}}_{\ell m} t^{\mbox{c}}_{\ell s} (\bo{t}^{\mbox{c}}_{\ell})^\top \cov^{-1} (\bo{t}^{\mbox{c}}_{\ell}) \Bigg], m,s \in\{1,\ldots, q\},$$

$$ \displaystyle (\cov^{-1}\cov_4)_{ms} 
= \frac{1}{q+2} \Bigg[ \sum_{\ell=1}^{k} \alpha_{\ell} (\bo{t}^{\mbox{c}}_{\ell})^\top \cov^{-1} (\bo{t}^{\mbox{c}}_{\ell}) \sum_{i=1}^{q} b_{mi} t^{\mbox{c}}_{\ell i} t^{\mbox{c}}_{\ell s} \Bigg], m,s \in\{1,\ldots, q\}.$$

\section*{Appendix B. \, Details for the study of \texorpdfstring{$\cov^{-1}\cov_4$}{} eigenvalues when group proportions vary in the absence of within-group variability  in Section~\ref{sec:dprop}}

In Section~\ref{sec:dprop}, we consider the case of a Dirac mixture, the scatter pair $\cov-\cov_4$, and $p=q=k-1$, where $p$ is the number of variables and $q$ is the dimension of the space spanned by the group centers. We know that in this case the group centers have no effect on the eigenvalues of $\cov-\cov_4$. Therefore, we only vary the proportions of the groups. The eigenvalues of $\cov^{-1}\cov_4$ are obtained with numerical calculations using the theoretical developments described in Appendix A.2.

\subsection*{Appendix B.1. \, The case of three groups}

Subsection~\ref{subsec:dprop3} focuses on the case where the number of groups $k$ is 3, and thus the number of variables $p$ is 2. The eigenvalues are represented in ternary diagrams in Fig.~\ref{fig:tern_grad} and \ref{fig:tern_plot}. The construction of such figures is achieved through the generation of a grid comprising all possible combinations of the group proportions $\alpha_1$, $\alpha_2$ and $\alpha_3$, with values between zero and one, such that their sum is equal to one. The step size is 0.001 (0.1\%). This approach ensures exhaustive coverage of the proportion space and guarantees that every point on the ternary diagram represents a valid combination of $\alpha_1$, $\alpha_2$, and $\alpha_3$.
The eigenvalues of $\cov^{-1}\cov_4$, $\rho_1$ and $\rho_2$, are calculated for each combination of the grid. To perform this calculation, it is necessary to have values for the centers of the three groups. As these values will have no impact on the result, they may be assigned randomly. For the sake of simplicity, we choose the first row of Table S1.

In Fig.~\ref{fig:tern_grad}, it is of interest to highlight the case in which the eigenvalue is equal to one (white), as this is the case in which ICS does not ``work''. Subsequently, two situations are distinguished: when the eigenvalue is smaller than one (blue), and when the eigenvalue is greater than one (red). Both blue and red parts follow a color gradient to illustrate the distance of the value from one. The logarithm of the eigenvalue was used in the coloring of the plot, because it allows for a more balanced distribution of colors. The logarithmic scale reduces the range difference between the blue and red parts, and since $log(1) = 0$, it conveniently positions the transition color at a meaningful point. The scale of the gradient ranges from blue at the minimum of $log(\rho)$ to red at the maximum of $log(\rho)$, $log(\rho)=0$ being white. The gradient effectively illustrates the variations in the eigenvalue across the group proportions grid and highlights instances where ICS does not ``work''. However, the color bar at the bottom of the plot represents the original eigenvalue, enhancing the plot's interpretability.

Both eigenvalues are shown qualitatively in Fig.~\ref{fig:tern_plot}.  Since the information of interest is whether the eigenvalues are greater or less than one (as explained in Subsection \ref{subsec:dprop3}, three cases are distinguished: $\rho_1$ and $\rho_2 > 1$  (red), $\rho_1 > 1$ and $\rho_2 < 1 $ (purple), and $\rho_1$ and $\rho_2 < 1$ (blue)).

\subsection*{Appendix B.2. \, The case of \texorpdfstring{$k$}{} groups}

Subsection~\ref{subsec:dpropk} generalizes the context to a number $k$ of groups and focuses on the notion of ``thresholds''. Fig.~\ref{fig:thresh} is an illustration of the thresholds found using three different setups, for 2 to 10 groups. The threshold values are described in Table \ref{tab:thresholds}.

\begin{table}[!ht]
    \caption{Group proportion thresholds computed for different values of the number of groups $k$ and three setups. Setup 1: all group proportions initially equal. Setup 2: $\alpha_2$ near the Setup 1 threshold. Setup 3: $\alpha_2$ equal to 5\%.}
    \label{tab:thresholds}
    \vskip0.3cm\hrule
    \smallskip
    \centering\small
    \begin{tabular}{cccc}
        k & Setup 1 & Setup 2 & Setup 3  \\ \hline
        3 & 0.18 & 0.2 & 0.171  \\ 
        4 & 0.155 & 0.161 & 0.151  \\ 
        5 & 0.135 & 0.138 & 0.133  \\ 
        6 & 0.12 & 0.121 & 0.119  \\ 
        7 & 0.107 & 0.108 & 0.107  \\ 
        8 & 0.097 & 0.097 & 0.097  \\ 
        9 & 0.089 & 0.089 & 0.088  \\ 
        10 & 0.082 & 0.082 & 0.081
    \end{tabular}
    \hrule
\end{table}

The first setup (Setup~1 in Fig.~\ref{fig:thresh}) is the simplest. It is known that the sum of all group proportions must be equal to one. The most balanced scenario is the one in which all groups have the same proportion. In this case, each proportion is equal to $1/k$, where $k$ is the number of groups. Thus, the smallest group proportion cannot exceed $1/k$ and the largest group proportion cannot be strictly smaller than $1/k$.
A grid with group proportions is created for a given number of groups. The proportions are ordered in ascending order. For the first group proportion $\alpha_1$, the values range from 0.001 to$ 1/k$ with a step of 0.001. For each intermediate group proportion $\alpha_j$, $j \in\{2,\ldots, k-1\}$, the value is set to $1/k$. The last group proportion $\alpha_k$ is the remaining proportion: $\alpha_k = 1 - \sum_{j=1}^{k-1}\alpha_j$. This setup ensures that the sum of group proportions is exactly one. Table \ref{tab:combined_grids_eg} provides an example of the grid obtained with the first setup when $k=4$. For each scenario in the grid, the eigenvalues of $\cov^{-1}\cov_4$ are calculated. Initially, all proportions are equal to $1/k$, and all eigenvalues are less than one. Subsequently, the proportion of the first group is the smallest and decreases by 0.001 at each iteration. The threshold of Setup~1 is the first value of the first group proportion for which one eigenvalue exceeds one. This procedure is repeated for 2 to 10 groups.

\begin{table}[!ht]
    \caption{Group proportions ($\alpha_1, \ldots, \alpha_4$) used in Setups 1, 2, and 3 for $k = 4$ groups. Setup 1: all group proportions initially equal. Setup 2: $\alpha_2$ near the Setup 1 threshold. Setup 3: $\alpha_2$ equal to 5\%.}
    \label{tab:combined_grids_eg}
    \vskip0.3cm\hrule
    \smallskip
    \centering\small
    \begin{tabular}{cccc@{\hspace{1cm}}cccc@{\hspace{1cm}}cccc}
        \multicolumn{4}{c}{Setup 1} & \multicolumn{4}{c}{Setup 2} & \multicolumn{4}{c}{Setup 3} \\ \hline
        $\alpha_1$ & $\alpha_2$ & $\alpha_3$ & $\alpha_4$ & $\alpha_1$ & $\alpha_2$ & $\alpha_3$ & $\alpha_4$ & $\alpha_1$ & $\alpha_2$ & $\alpha_3$ & $\alpha_4$ \\ \hline
        0.25 & 0.25 & 0.25 & 0.25 & 0.25 & 0.175 & 0.25 & 0.325 & 0.25 & 0.05 & 0.25 & 0.45 \\ 
        0.249 & 0.25 & 0.25 & 0.251 & 0.249 & 0.175 & 0.25 & 0.326 & 0.249 & 0.05 & 0.25 & 0.451 \\ 
        0.248 & 0.25 & 0.25 & 0.252 & 0.248 & 0.175 & 0.25 & 0.327 & 0.248 & 0.05 & 0.25 & 0.452 \\ 
        \vdots & \vdots & \vdots & \vdots & \vdots & \vdots & \vdots & \vdots & \vdots & \vdots & \vdots & \vdots \\ 
        0.001 & 0.25 & 0.25 & 0.499 & 0.001 & 0.175 & 0.25 & 0.574 & 0.001 & 0.05 & 0.25 & 0.699
    \end{tabular}
    \hrule
\end{table}

The second setup (Setup~2 in Fig.~\ref{fig:thresh}) uses a grid based on the results obtained with the first setup. From the Setup~1 grid, the value of the second group proportion is updated. It is set to the threshold found in Setup~1 to which 2\% is added. The proportion of the first group is the same, ranging from 0.001 to $1/k$ with a step of 0.001. The intermediate groups have a proportion of $1/k$ only for group 3 through $k-1$, since the proportion of the second group is already set. The value of the last group proportion is also updated, but it is still equal to one minus the other group weights. This procedure ensures that the sum of each row is exactly one. An example of the grid obtained with the second setup when $k=4$ is shown in Table \ref{tab:combined_grids_eg}.  However, it requires at least three groups: one that is set to the first threshold plus 2\%, and two groups with variable proportions. For each configuration of group weights, the eigenvalues of $\cov^{-1}\cov_4$ are calculated. Initially, the first group is not necessarily the smallest one; it could be the second group. However, the value of the second group proportion is set such that it is above the threshold, implying that all eigenvalues are initially below one. Thus, the threshold of Setup~2 is, as in Setup~1, the first value of the first group weight for which one eigenvalue exceeds one.
The concept in this procedure revolves around examining the impact of a group that is positioned near the threshold to determine if this proximity affects it. This process is repeated for a number of groups ranging from 3 to 10.

Finally, the third procedure (Setup~3 in Fig.~\ref{fig:thresh}) uses a grid similar to Setup~2. The second group proportion is no longer dependent on the threshold of Setup~1 but is set to 0.05. Again, the last group proportion is updated to one minus the combined weights of the other groups.  All other group proportions remain the same as in Setups 1 and 2. Table \ref{tab:combined_grids_eg} provides an example of the third setup grid with $k=4$. As for the grid in Setup~2, it requires at least three groups: one for the fixed proportion and two to vary. The objective of this procedure is to include a small group initially, specifically below the threshold for $k$ from 3 to 10. This implies that, from the initial point, there is already an eigenvalue greater than one. Thus, the threshold of Setup~3 is the first value of the first group weight for which two eigenvalues exceed one. This is aimed at identifying the point where the second eigenvalue surpasses one. This setup is repeated for 3 to 10 groups.

\section*{Appendix C. \,  Proof of Proposition \ref{prop:k3q1}}
\begin{proof}[\textbf{\upshape Proof:}]
By adapting the calculations in Appendix A (with the same notations) to the Gaussian mixture model \eqref{modelk3q1}, we obtain:
\begin{align*}
\cov = \begin{bmatrix} \beta_{11} & \bo{0} \\ \bo{0} & \bo{I}_{p-1} \end{bmatrix}, \quad \cov^{-1} = \begin{bmatrix} b_{11} & \bo{0} \\ \bo{0} & \bo{I}_{p-1} \end{bmatrix}, \quad \cov_4 &= \begin{bmatrix} a_{11} & \bo{0} \\ \bo{0} & \bo{I}_{p-1}, \end{bmatrix}, 
\end{align*}
where $\beta_{11} = 1 + \sum_{j=1}^{3} \alpha_j (t_{j1}^{\mbox{c}})^2$, $b_{11} = \beta_{11}^{-1}$ and $ \displaystyle a_{11} =\frac{1}{p+2}\Bigg[b_{11}\E[(x^{\mbox{c}}_1)^4] +(p-1) \E[(x^{\mbox{c}}_1)^2]\Bigg]$. We have:
\begin{align*}
    \cov^{-1}\cov_4 = \begin{bmatrix} b_{11} a_{11} & \bo{0} \\ \bo{0} & \bo{I}_{p-1} \end{bmatrix}.
\end{align*}
All eigenvalues of $\cov^{-1}\cov_4$ equal one if and only if
\begin{equation}\label{eq:eigk3q1}
    b_{11} a_{11}=1,
\end{equation}
\eqref{eq:eigk3q1} is equivalent to $a_{11}=b_{11}^{-1}=\beta_{11}$. Noting that $\E[(x^{\mbox{c}}_1)^2]= \beta_{11}$, and expanding the expression of $a_{11}$, \eqref{eq:eigk3q1} is equivalent to $3\beta_{11}^2 = \E[(x^{\mbox{c}}_1)^4]$. Using the formula of $\beta_{11}$ and the formula of $\E[(x^{\mbox{c}}_1)^4]$ from Appendix A, we obtain that \eqref{eq:eigk3q1} is equivalent to:

\begin{equation*}
\begin{aligned}
& \alpha_1(3\alpha_1-1)(t_{11}^{\mbox{c}})^4 +\alpha_2(3\alpha_2-1)(t_{21}^{\mbox{c}})^4 + \alpha_3(3\alpha_3-1)(t_{31}^{\mbox{c}})^4 \\
&\quad +6\alpha_1\alpha_2(t_{11}^{\mbox{c}})^2(t_{21}^{\mbox{c}})^2 + 6\alpha_1\alpha_3(t_{11}^{\mbox{c}})^2(t_{31}^{\mbox{c}})^2 +6\alpha_2\alpha_3(t_{21}^{\mbox{c}})^2(t_{31}^{\mbox{c}})^2 = 0.
\end{aligned}
\end{equation*}

For the mixture \eqref{modelk3q1}, we know from Theorem 4 in \cite{tyler_invariant_2009} and using the computations of Subsection \ref{subsec:cc4calc}, that at least $p-1$ eigenvalues of $\cov^{-1}\cov_4$ are equal to one.
As for the remaining eigenvalue, using the computations of Subsection \ref{subsec:cc4calc} and Mathematica \cite{inc_mathematica_nodate}, we can prove that it is equal to one if and only if 
\begin{equation*}\label{eq:peq0}
    p_{\alpha_1, \alpha_2}(t_{11}, t_{12})=0,
\end{equation*}
where
\begin{equation*}
\begin{aligned}
   p_{\alpha_1, \alpha_2}(t_{11}, t_{12})&= \alpha_1 (-1 + 7 \alpha_1 - 12 \alpha_1^2 + 6 \alpha_1^3) t_{11}^4 + 4 \alpha_1 \alpha_2(1 - 6 \alpha_1 + 6 \alpha_1^2) t_{11}^3 t_{12} +
 6 \alpha_1 \alpha_2 (1 - 2 \alpha_2 \\ 
 & + \alpha_1 (-2 + 6 \alpha_2)) t_{11}^2 t_{12}^2 + 
 4 \alpha_1 \alpha_2 (1 - 6 \alpha_2 + 6 \alpha_2^2) t_{11} t_{12}^3 + \alpha_2 (-1 + 7 \alpha_2 - 12 \alpha_2^2 + 6 \alpha_2^3) t_{12}^4.
\end{aligned}
\end{equation*}
Since $t_{21} \neq 0$, we have:
\begin{align*}
    p_{\alpha_1, \alpha_2}(t_{11}, t_{12}) 
    &= t_{21}^4 r_{\alpha_1, \alpha_2}(t_{11}/t_{21}),
\end{align*}
for the polynomial  $r_{\alpha_1, \alpha_2}$ of degree 4 given in Proposition \ref{prop:k3q1}, which concludes the proof.
\end{proof}

\section*{Appendix D. \, Supplementary data}
Supplementary material related to this article can be found online at (to be filled by the journal).

\bibliographystyle{unsrtnat}
\bibliography{references_R2}

\end{document}